%% file: paper.tex
\documentclass[a4paper,UKenglish,cleveref,autoref,thm-restate,dvipsnames]{lipics-v2021}

\usepackage{tikz}
\usetikzlibrary{intersections, calc, arrows, automata, positioning,shapes.multipart} 
\tikzstyle{BoxStyle} = [draw, circle, fill=black, scale=0.4,minimum width = 1pt, minimum height = 1pt]

\usepackage{hyperref}
\usepackage{graphicx}

\usepackage[linesnumbered,ruled]{algorithm2e}
\usepackage{amsmath}
\allowdisplaybreaks
\DeclareMathOperator*{\argmin}{arg\,min}

\usepackage{floatrow}
\newfloatcommand{capbtabbox}{table}[][\FBwidth]
\usepackage{tabularx}

\usepackage{colortbl}
\usepackage{longtable}

\usepackage[font=small,labelfont=bf]{caption}
\usepackage{subcaption}

\usepackage{mathtools}
\usepackage{amssymb}
\usepackage{stmaryrd}

\usepackage{bm}
\newcommand{\lfp}{\bm{\mu}}

\usepackage{amsthm}

\newcommand{\commenteq}[1]{\hspace{2em} [\mbox{#1}]}

\newcommand{\suchthat}{\;\ifnum\currentgrouptype=16 \middle\fi|\;}

\usepackage[capitalise]{cleveref}

\mathcode`<="4268
\mathcode`>="5269
\mathchardef\gr="213E
\mathchardef\ls="213C

\newcommand{\Dist}{\mathrm{Distr}}

\newcommand{\support}{\mathrm{support}}

\newcommand{\nat}{\mathbb{N}}

\newcommand{\A}{\mathcal{A}}
\newcommand{\Act}{Act}

\newcommand{\M}{\mathcal{M}}
\newcommand{\D}{\mathcal{D}}
\newcommand{\C}{\mathcal{C}}

\newcommand{\SetP}{\mathcal{P}}

\newcommand{\pfun}{\mathrel{\ooalign{\hfil$\mapstochar\mkern5mu$\hfil\cr$\to$\cr}}}

\newcommand{\m}{{\sf m}}
\newcommand{\Paths}{{\sf Paths}}
\newcommand{\last}{{\sf last}}
\newcommand{\R}{\mathcal{R}}

\newcommand{\QT}[1]{{\color{black}#1}}
\newcommand{\SK}[1]{{\color{black}#1}}

\newcommand{\ETR}{\sf \displaystyle \exists \mathbb{R}}

\pdfoutput=1 
\hideLIPIcs  

\title{Minimising the Probabilistic Bisimilarity Distance} 

\author{Stefan Kiefer}{Department of Computer Science, University of Oxford, UK}{stekie@cs.ox.ac.uk}{https://orcid.org/0000-0003-4173-6877}{}

\author{Qiyi Tang}{Department of Computer Science, University of Liverpool, UK}{qiyi.tang@liverpool.ac.uk}{https://orcid.org/0000-0002-9265-3011}{}

\authorrunning{S.\ Kiefer and Q.\ Tang} 

\Copyright{Stefan Kiefer and Qiyi Tang} 



\category{} 

\relatedversion{The paper is accepted to CONCUR 2024.} 

\funding{This work has been supported in part by the Engineering and Physical
	Sciences Research Council  (EPSRC) through grant EP/X042596/1.}

\acknowledgements{We thank the referees for their constructive feedback.}

\nolinenumbers 

\EventEditors{Rupak Majumdar and Alexandra Silva}
\EventNoEds{2}
\EventLongTitle{35th International Conference on Concurrency Theory (CONCUR 2024)}
\EventShortTitle{CONCUR 2024}
\EventAcronym{CONCUR}
\EventYear{2024}
\EventDate{September 9--13, 2024}
\EventLocation{Calgary, Canada}
\EventLogo{}
\SeriesVolume{311}
\ArticleNo{5}

\begin{document}\sloppy
	
	\maketitle
	
	\begin{abstract}
		A labelled Markov decision process (MDP) is a labelled Markov chain with nondeterminism; i.e., together with a strategy a labelled MDP induces a labelled Markov chain. The model is related to interval Markov chains.
		Motivated by applications to the verification of probabilistic noninterference in security, we study problems of minimising probabilistic bisimilarity distances of labelled MDPs, in particular, whether there exist strategies such that the probabilistic bisimilarity distance between the induced labelled Markov chains is less than a given rational number, both for memoryless strategies and general strategies.
		We show that the distance minimisation problem is $\ETR$-complete for memoryless strategies and undecidable for general strategies.
		We also study the computational complexity of the qualitative problem about making the distance less than one.
		This problem is known to be NP-complete for memoryless strategies.
		We show that it is EXPTIME-complete for general strategies.	
	\end{abstract}
	
	\section{Introduction} \label{section:introduction}
	Given a model of computation (e.g., finite automata), and two instances of it, are they semantically equivalent (e.g., do the automata accept the same language)?
	Such \emph{equivalence} problems can be viewed as a fundamental question for almost any model of computation.
	As such, they permeate computer science, in particular, theoretical computer science.
	
	In \emph{labelled Markov chains (LMCs)}, which are Markov chains whose states (or, equivalently, transitions) are labelled with an observable letter, there are two natural and very well-studied versions of equivalence, namely \emph{trace (or language) equivalence} and \emph{probabilistic bisimilarity}.
	
	The \emph{trace equivalence} problem has a long history, going back to Sch\"utzenberger~\cite{Schutzenberger} and Paz~\cite{Paz2014}
	who studied \emph{weighted} and \emph{probabilistic} automata, respectively.
	Those models generalise LMCs, but the respective equivalence problems are essentially the same.
	For LMCs, trace equivalence asks if the same label sequences have the same probabilities in the two LMCs.
	It can be extracted from~\cite{Schutzenberger} that equivalence is decidable in polynomial time, using a technique based on linear algebra; see also \cite{Tzeng,DoyenHR08}.
	
	\emph{Probabilistic bisimilarity} is an equivalence that was introduced by Larsen and Skou~\cite{LarsenS91}.
	It is finer than trace equivalence, i.e., probabilistic bisimilarity implies trace equivalence.
	A similar notion for Markov chains, called \emph{lumpability}, can be traced back at least to the classical text by Kemeny and Snell~\cite{KemenySnell60}.
	Probabilistic bisimilarity can also be computed in polynomial time \cite{Baier1996,DerisaviHS03,ValmariF10}.
	Indeed, in practice, computing the bisimilarity quotient is fast and has become a backbone for highly efficient tools for probabilistic verification such as \textsc{Prism}~\cite{Prism} and \textsc{Storm}~\cite{Storm}.
	
	Numerous quantitative generalisations of this
	behavioural equivalence have been proposed, the probabilistic bisimilarity distance due to
	Desharnais et al. \cite{DGJP1999} being the most notable one. 
	This distance can be at most 1, and is 0 if and only if the LMCs are probabilistic bisimilar.
	It was shown in \cite{CvBW2012} that the distance can be computed in polynomial time.
	
	In this paper, we study distance minimisation problems for \emph{(labelled) Markov decision processes (MDPs)}, which are LMCs plus nondeterminism; i.e., each state may have several \emph{actions} (or ``moves'') one of which is chosen by a controller, potentially randomly.
	An MDP and a controller \emph{strategy} together induce an LMC (potentially with infinite state space, depending on the complexity of the strategy).
	We consider both general strategies and the more restricted memoryless ones.
	There are good reasons to consider memoryless strategies, particularly their naturalness and simplicity in implementations, and their connection to \emph{interval Markov chains} (see, e.g., \cite{JonssonL91,Delahaye15}) and \emph{parametric MDPs} (see, e.g., \cite{HahnHZ11,WinklerJPK19}).
	There are also good reasons to consider general unrestricted strategies, primarily their naturalness (in their definition for MDPs) and their generality. The latter is important particularly for security applications, see below, where general strategies can make programs more secure, in a precise, quantitative sense.
	
	Let us elaborate on the connection to security.
	\emph{Noninterference} refers to an information-flow property of a program, stipulating that information about \emph{high} data (i.e., data with high confidentiality) may not leak to \emph{low} (i.e., observable) data, or, quoting~\cite{SabelfeldS00}, ``that a program is secure whenever varying the initial values of high variables cannot change the low-observable (observable by the attacker) behaviour of the program''.
	It was proposed in~\cite{SabelfeldS00} to reason about \emph{probabilistic} noninterference in probabilistic multi-threaded programs by proving probabilistic bisimilarity; see also \cite{Smith03,PopescuHN13}.
	More precisely, probabilistic noninterference is established if it can be shown that any two states that differ only in high data are probabilistic bisimilar, as then an attacker who only observes the low part of a state learns nothing about the high part.
	The observable behaviour of a multi-threaded program depends strongly on the \emph{scheduler},
	which in this context amounts to a strategy in the corresponding MDP.
	
	Nevertheless, ensuring
	perfect (probabilistic) noninterference
	proves challenging, and a certain degree of information leakage may be acceptable \cite{Gray1990, RMMG2001}.
	In such scenarios, where (probabilistic) bisimilarity might not hold under any scheduler, turning to bisimilarity distances allows us to estimate the security degree of a system under different schedulers.
	The smaller the distance, the more secure the system.
	Therefore, we would like to devise schedulers that minimise the probabilistic bisimilarity distances.
	
	
	
	Some qualitative problems have already been studied in previous work.
	Concerning memoryless strategies, it was shown in \cite{KT2020} that the bisimilarity equivalence problem, i.e., whether strategies exist to make the distance 0, is NP-complete.
	Similarly, it was also shown in \cite{KT2020} that the problem whether memoryless strategies exist to make the distance less than one is NP-complete; cf.\ \cref{tab:summary}.
	The bisimilarity \emph{in}equivalence problem, i.e., whether strategies exist to make the distance greater than 0, can be decided in polynomial time for memoryless strategies~\cite{KT2020}. 
	
	Concerning general strategies, the bisimilarity equivalence and inequivalence problems were studied in \cite{KT2022}.
	It was shown there that these problems are EXPTIME-complete and in~P, respectively.
	
	It remained open whether the existence of strategies to make the distance less than one is decidable for general strategies.
	We show that the distance less than one problem for general strategies is decidable.
	In fact, it is EXPTIME-complete, and therefore the problem has the same complexity as the bisimilarity equivalence problem for general strategies.
	To obtain this result, we prove a tight connection between the distance less than one problem and the bisimilarity equivalence problem:
	loosely speaking, whenever there are general strategies for two states to have distance less than one,
	the two states can reach a pair of states whose distance can be made 0, thus probabilistic bisimilar.
	This connection is natural and known for \emph{finite} labelled Markov chains, but nontrivial to establish in general.

	We also study \emph{quantitative} distance minimisation problems: do there exist memoryless (resp.\ general) strategies for two given MDPs such that the induced LMCs have distance less than a given threshold?
	We show that the distance minimisation problem is $\ETR$-complete for memoryless strategies and undecidable for general strategies.
	Here, $\ETR$ refers to the class of problems that are many-one reducible to the existential theory of the reals; it is known that ${\sf NP} \subseteq \ETR \subseteq {\sf PSPACE}$.
	
	\begin{table}
		\centering
		\begin{tabular}{|c|c|c|}
			\hline
			Problem & Memoryless Strategy & General Strategy\\
			\hline \hline
			distance $= 0$    & NP-complete \cite{KT2020} &  EXPTIME-complete \cite{KT2022}\\
			distance $\ls \, 1$  & NP-complete \cite{KT2020}  & EXPTIME-complete (\cref{section:distance-less-than-one})\\
			distance $\ls \, \theta$  & $\ETR$-complete (\cref{section:memoryless-strategies})  & undecidable (\cref{section:general-strategies})\\
			\hline
		\end{tabular}	
		
		\caption{Summary of results on distance minimisation problems.}
		\label{tab:summary}
	\end{table}
	
	The rest of the paper is organised as follows.
	We give preliminaries in \cref{section:preliminaries}.
	In \cref{section:noninterference} we discuss probabilistic noninterference.
	In \cref{section:memoryless-strategies,section:general-strategies} we prove our results on the quantitative distance minimisation problems for general strategies and memoryless strategies, respectively.
	We study the distance less than one problem for general strategies in \cref{section:distance-less-than-one}.
	We conclude in \cref{section:conclusion}.
	Missing proofs can be found in an appendix.
	
	\section{Preliminaries} \label{section:preliminaries}
	We write $\nat$ for the set of nonnegative integers.
	Let $S$ be a finite set.
	We denote by $\Dist(S)$ the set of probability distributions on~$S$.
	For a distribution $\mu \in \Dist(S)$ we write $\support(\mu) = \{s \in S \suchthat \mu(s) \gr 0 \}$ for its support.
	We denote the Dirac distribution concentrated on an element $s \in S$ by $\mathbf{1}_s$, that is, $\mathbf{1}_s(s) =1$ and $\mathbf{1}_s(t) =0$ for all $t \neq s$.
	We denote by $\rho(i)$ the $i$-th element of a sequence $\rho$.
	We denote the least fixed point of a function $f$ by $\lfp.f$.
	
	A \emph{labelled Markov chain} (LMC) is a quadruple $<S, L, \tau, \ell>$ consisting of a nonempty countable set $S$ of states, a nonempty finite set $L$ of labels, a transition function $\tau : S \to \Dist(S)$, and a labelling function $\ell: S \to L$.
	We denote by $\tau(s)(t)$ the transition probability from $s$ to $t$.
	Similarly, we denote by $\tau(s)(E) = \sum_{t \in E} \tau(s)(t)$ the transition probability from $s$ to $E \subseteq S$.
	We require the LMCs to be finitely branching, that is, $|\support(\tau(s))|$ is finite for every $s \in S$.
	
	An equivalence relation $R \subseteq S \times S$ is a \emph{probabilistic bisimulation} if for all $(s, t) \in R$, $\ell(s) = \ell(t)$ and $\tau(s)(E) = \tau(t)(E)$ for each $R$-equivalence class $E$.
	\emph{Probabilistic bisimilarity}, denoted by $\mathord{\sim_{\M}}$ (or $\mathord{\sim}$ when $\M$ is clear), is the largest probabilistic bisimulation.
	
	\QT{The probabilistic bisimilarity distance, a pseudometric on LMCs, was first defined by Desharnais, Gupta, Jagadeesan and Panangaden in \cite{DGJP1999}. Their definition is based on a real-valued modal logic. This logic can be viewed as a function which maps a formula~$f$ of the logic and a state~$s$ of the LMC to a real number $f(s) \in [0, 1]$. The distance $d(s,t)$ between two states $s,t$ is defined as $\sup_f |f(s)-f(t)|$. Later, Van Breugel and Worrell \cite{vBW2001} defined probabilistic bisimilarity distances on LMCs as a fixed point of a function. They showed that their pseudometric coincides with the one defined in \cite{DGJP1999}.
		In this paper, we use the definition from~\cite{vBW2001}.
		The \emph{probabilistic bisimilarity distance}, denoted by $d_{\M}$} (or $d$ when $\M$ is clear), is a function from $S \times S$ to $[0, 1]$, that is, an element of $[0, 1]^{S \times S}$.
	It \SK{is} the least fixed point of the following function:
	\[
	\Delta(e)(s, t) = \left \{
	\begin{array}{ll}
		1 & \mbox{if $\ell(s) \not= \ell(t)$}\\
		\displaystyle \min_{\omega \in \Omega(\tau(s), \tau(t))} \sum_{u, v \in S} \omega(u, v) \; e(u, v) & \mbox{otherwise}
	\end{array}
	\right .
	\]
	where the set $\Omega(\mu, \nu)$ of \emph{couplings} of $\mu,\nu \in \Dist(S)$ is defined as $\Omega(\mu, \nu)=\left \{ \, \omega \in \Dist(S \times S) \suchthat \sum_{t \in S} \omega(s, t) = \mu(s) \wedge \sum_{s \in S} \omega(s, t) = \nu(t) \, \right \}$. Note that a coupling $\omega \in \Omega(\mu, \nu)$ is a joint probability distribution with marginals $\mu$ and $\nu$ (see, e.g., \cite[page 260-262]{Billingsley1995}). For all $s,t \in S$, $s \sim t$ if and only if $s$ and $t$ has probabilistic bisimilarity distance zero \cite[Theorem~1]{DGJP1999}.

	
	A \emph{(labelled) Markov decision process} (MDP) is a tuple $<S, \Act, L, \varphi, \ell>$ consisting of a finite set $S$ of states, a  finite set $\Act$ of actions, a  finite set $L$ of labels, a partial function $\varphi: S \times \Act \pfun \Dist(S)$ denoting the probabilistic transition, and a labelling function $\ell: S \to L$. The set of available actions in a state $s$ is $\Act(s) = \{\m \in \Act \mid \varphi(s, \m) \text{ is defined}\}$.
	
	A path is a sequence $\rho = s_0\m_1s_1\cdots\m_ns_n$ such that $\varphi(s_i, \m_{i+1})$ is defined and $\varphi(s_i, \m_{i+1})(s_{i+1}) \gr 0$ for all $0 \le i \ls n$. The last state of $\rho$ is $\last(\rho) = s_n$. Let $\Paths(\D)$ denote the set of paths in $\D$.
	
	A (general) strategy for an MDP is a  function $\alpha: \Paths(\D) \to \Dist(\Act)$ that given a path $\rho$, returns a probability distribution on the available actions at the last state of $\rho$, $\last(\rho)$.
	A memoryless strategy depends only on $\last(\rho)$; so we can identify a memoryless strategy with a function $\alpha: S \to\Dist(\Act)$ that given a state $s$, returns a probability distribution on the available actions at that state.
	
	A general strategy $\alpha$ for $\D$ induces an LMC $\D(\alpha) = <\SetP, L,\tau,\ell'>$, where $\SetP \subseteq \Paths(\D)$. For $\rho \in \SetP$, we have $\tau(\rho)(\rho \m t) = \alpha(\rho)(\m)\varphi(s, \m)(t)$ and $\ell'(\rho) = \ell(s)$ where $s = \last(\rho)$ and $\m \in \Act(s)$.
	
	
	\input{noninterference}

	\section{Memoryless Strategies: Distance Minimisation}
	\label{section:memoryless-strategies}
	
	In this section we consider the \emph{memoryless distance minimisation problem} which, given an MDP, two states $s_1, s_2$ of the MDP, and a rational number $\theta$, asks whether there is a memoryless strategy~$\alpha$ such that $d(s_1,s_2) \ls \theta$ holds in the LMC induced by~$\alpha$.

	We show that the memoryless distance minimisation problem is $\ETR$-complete.
	We prove the lower and upper bound in \cref{thm:memmin,thm:memmin-upper}, respectively.
	
	The \emph{existential theory of the reals, {\sf ETR},} is the set of valid formulas of the form
	$$\exists x_1 \dots \exists x_n~R(x_1, \dots, x_n),$$
	where $R$ is a Boolean combination of comparisons of the form
	$p(x_1, \dots, x_n) \sim 0$, in which $p(x_1, \dots, x_n)$ is a multivariate polynomial (with rational coefficients) and
	$\mathord{\sim} \in \{ \ls, \gr, \mathord{\le}, \mathord{\ge}, \mathord{=}, \mathord{\ne} \}$.
	The complexity class $\ETR$~\cite{SchaeferS17} consists of those problems that are many-one reducible to {\sf ETR} in polynomial time.
	Since {\sf ETR} is {\sf NP}-hard and in {\sf PSPACE}~\cite{Can88,Renegar92}, we have ${\sf NP} \subseteq \ETR \subseteq {\sf PSPACE}$.
	
	To prove that the memoryless distance minimisation problem is $\ETR$-hard (\cref{thm:memmin}), we proceed via a sequence of reductions, represented by the following lemmas, \cref{lem:ETR-1,lem:ETR-2,lem:ETR-3}.
	
	\begin{lemma} \label{lem:ETR-1}
		The following problem is $\ETR$-complete: given a multivariate polynomial $p : \mathbb{R}^n \to \mathbb{R}$ of (total) degree at most~$6$, does there exist $x \in \mathbb{R}^n$ with $p(x) \ls 0$?
		The problem remains $\ETR$-complete under the promise that if there is $x$ with $p(x) \ls 0$ then there is $x'$ with $p(x') \ls 0$ and $\|x'\| \ls 1$ (where $\| \cdot \|$ denotes the Euclidean norm).
	\end{lemma}
	
	\begin{proof}
			Membership in~$\ETR$ is clear.
		It remains to prove $\ETR$-hardness.
		It is shown in \cite[Lemma~3.9]{Schaefer13} that the following problem is $\ETR$-complete: given multivariate polynomials $f_1, \ldots, f_s: \mathbb{R}^n \to \mathbb{R}$, each of degree at most~$2$, does there exist $x \in \mathbb{R}^n$ with $\|x\| \ls 1$ such that $\bigwedge_{i=1}^s (f_i(x) = 0)$?
		It follows from the proof that the problem remains $\ETR$-complete under the promise that $\bigwedge_i f_i(x) = 0$ implies $\|x\| \ls 1$.
		We reduce from this promise problem.
		Let $f_1, \ldots, f_s: \mathbb{R}^n \to \mathbb{R}$, each of degree at most~$2$, such that for all $x \in \mathbb{R}^n$ we have that $\bigwedge_i f_i(x) = 0$ implies $\|x\| \ls 1$.
		Define the polynomial $q : \mathbb{R}^n \to \mathbb{R}$ by $q(x) := \sum_{i=1}^s f_i(x)^2$.
		Clearly, $q(x) \ge 0$ always holds, and we have $q(x) = 0$ if and only if $\bigwedge_i f_i(x) = 0$.
		Consider the two sets $\{(q(x),x) \in \mathbb{R}^{n+1} \mid \|x\| \le 1\}$ and $\{(0,x) \in \mathbb{R}^{n+1} \mid \|x\| \le 1\}$.
		If $q$ has a root~$x$, then the two sets overlap in the point $(0,x)$; otherwise, by \cite[Corollary~3.4]{SchaeferS17}, they have distance at least $2^{2^{-k}}$, where $k$~is a natural number whose unary representation can be computed in polynomial time.
		It follows that if $\|x\| \le 1$ and $q(x) \ls 2^{2^{-k}}$ then \QT{there exists $x'$ such that $q(x') = 0$}.
		
		In the following let us use real-valued variables $x_1, \ldots, x_n, y_{1}, \ldots, y_{k}$ and write $x = (x_1, \ldots, x_n)$ and $y = (y_1, \ldots, y_k)$.
		Define the polynomial $r : \mathbb{R}^{n+k} \to \mathbb{R}$ (of degree at most~$6$) by
		\[
		r(x,y) \ := \ (y_{1} - 4)^2 + (y_{2} - y_{1}^2)^2 + \cdots + (y_{k} - y_{k-1}^2)^2 + y_{k}^2 q(x) + \|x\|^2 - 1 \,.
		\]
		Let us also use a real-valued variable~$z$.
		Define the polynomial $p : \mathbb{R}^{n+k+1}$ (of degree at most~$6$) by
		\[
		p(x,y,z) \ := \ z^6 r\left(\frac{x_1}{z}, \ldots, \frac{x_{n}}{z}, \frac{y_1}{z}, \ldots, \frac{y_k}{z}\right)\,.
		\]
		
		Suppose there is $x \in \mathbb{R}^n$ with $\bigwedge_i f_i(x) = 0$.
		Then $q(x) = 0$.
		For $1 \le i \le k$, set $y_{i} := 2^{2^{i}}$.
		Then $r(x,y) = \|x\|^2 - 1 \ls 0$.
		Set $z \gr 0$ small enough so that $z^2 \left( \|x\|^2 + \|y\|^2 + 1 \right) \ls 1$.
		For $1 \le i \le n$, set $x_i' := x_i z$.
		For $1 \le i \le k$, set $y_i' := y_i z$.
		Then $p(x',y',z) = z^6 r(x,y) \ls 0$ and $\|x'\|^2 + \|y'\|^2 + z^2 = z^2 \left( \|x\|^2 + \|y\|^2 + 1 \right) \ls 1$.
		
		Towards the other direction, suppose there is $(x',y',z) \in \mathbb{R}^{n+k+1}$ with $p(x',y',z) \ls 0$.
		Since $p$~is a polynomial, it is continuous.
		So we can assume without loss of generality that $z \ne 0$.
		For $1 \le i \le n$, set $x_i := x_i' / z$.
		For $1 \le i \le k$, set $y_i := y_i' / z$.
		Then $r(x,y) = p(x',y',z)/z^6 \ls 0$.
		This implies $y_k^2 q(x) \ls 1$ and $\|x\| \ls 1$.
		Using $r(x,y) \ls 0$, we show by induction that $y_i \ge 2^{2^{i-1}} + 1$ holds for all $i \in \{1, \ldots, k\}$.
		For the induction base ($i=1$) we have $(y_1-4)^2 \le 1$.
		Thus, $y_1 - 4 \ge -1$, and so $y_1 \ge 3 = 2^{2^{1-1}} + 1$.
		For the step ($1 \le i \le k-1$), suppose that $y_i \ge 2^{2^{i-1}} + 1$.
		Since $r(x,y) \ls 0$, we have $(y_{i+1} - y_i^2)^2 \le 1$, and so
		\[
		y_{i+1} \ \ge \ y_i^2 - 1 \ \ge \ (2^{2^{i-1}} + 1)^2 - 1  \ = \ 2^{2^i} + 2 \cdot 2^{2^{i-1}} \ \ge \ 2^{2^i} + 1\,.
		\]
		Hence, we have shown that $y_k \ge 2^{2^{k-1}} + 1 \gr 2^{2^{k-1}}$.
		It follows that $q(x) \ls 1/y_k^2 \ls 2^{-2^k}$.
		Since $\|x\| \ls 1$, it follows from the argument at the beginning that \QT{there exists $x'$ such that $q(x') = 0$ and so $\bigwedge_i f_i(x') = 0$}.
		
		This completes the hardness proof.
		Note that by combining the two directions, it follows that if there is $w \in \mathbb{R}^{n+k+1}$ with $p(w) \ls 0$, then there is $w' \in \mathbb{R}^{n+k+1}$ with $p(w') \ls 0$ and $\|w'\| \ls 1$, showing also $\ETR$-hardness of the promise version of the problem.
	\end{proof}

	\begin{lemma} \label{lem:ETR-2}
		The following problem is $\ETR$-complete: given a multivariate polynomial $p : \mathbb{R}^n \to \mathbb{R}$ of degree at most~$6$, does there exist $x \in [0,1]^n$ with $p(x) \gr 0$?
	\end{lemma}
	
	\begin{proof}
			Membership in~$\ETR$ is clear.
		For hardness we reduce from the promise problem from the previous lemma.
		Let $p : \mathbb{R}^n \to \mathbb{R}$ be a multivariate polynomial of degree at most~$6$ such that if there is $x \in \mathbb{R}^n$ with $p(x) \ls 0$ then there is $x' \in \mathbb{R}^n$ with $p(x') \ls 0$ and $\|x'\| \ls 1$.
		Define the polynomial $q : \mathbb{R}^{2 n} \to \mathbb{R}$ by $q(y_1, \ldots, y_n, z_1, \ldots, z_n) := -p(y_1 - z_1, \ldots, y_n - z_n)$.
		The degree of~$q$ is at most~$6$.
		We have to show that there is $x \in \mathbb{R}^n$ with $p(x) \ls 0$ if and only if there are $y_1, \ldots, y_n, z_1, \ldots, z_n \in [0,1]$ with $q(y_1, \ldots, y_n, z_1, \ldots, z_n) \gr 0$.
		
		Suppose there are $x_1, \ldots, x_n \in \mathbb{R}$ with $p(x_1, \ldots, x_n) \ls 0$.
		By the property of~$p$ we can assume that $x_1^2 + \cdots + x_n^2 \ls 1$.
		It follows that $x_i \in [-1,1]$ holds for all~$i$.
		For all~$i$ with $x_i \ge 0$ define $y_i := x_i$ and $z_i := 0$.
		For all~$i$ with $x_i \ls 0$ define $y_i := 0$ and $z_i := -x_i$.
		Then we have $x_i = y_i - z_i$ and $y_i,z_i \in [0,1]$ for all~$i$.
		Further,
		\[
		q(y_1, \ldots, y_n, z_1, \ldots, z_n) \ = \ -p(y_1 - z_1, \ldots, y_n - z_n) \ = \ - p(x_1, \ldots, x_n) \ \gr \ 0\,.
		\]
		
		Towards the other direction, suppose that there are $y_1, \ldots, y_n, z_1, \ldots, z_n \in [0,1]$ with $q(y_1, \ldots, y_n, z_1, \ldots, z_n) \gr 0$.
		For all $i$ define $x_i := y_i - z_i$.
		Then we have
		\[
		p(x_1, \ldots, x_n) \ = \ p(y_1 - z_1, \ldots, y_n - z_n) \ = \ - q(y_1, \ldots, y_n, z_1, \ldots, z_n) \ \ls \ 0\,,
		\]
		as required.
	\end{proof}
	
	\begin{lemma} \label{lem:ETR-3}
		The following problem is $\ETR$-complete: given a rational number $\theta \ge 0$ and a multivariate (degree-$6$) polynomial $p : \mathbb{R}^n \to \mathbb{R}$ of the form $p(x) = \sum_{j=1}^k f_j(x)$ where each $f_j(x_1, \ldots, x_n)$ is a product of a nonnegative coefficient and $6$~terms of the form $x_i$ or $(1-x_i)$, does there exist $x \in [0,1]^n$ with $p(x) \gr \theta$?
	\end{lemma}
	\begin{proof}
		Membership in~$\ETR$ is clear.
		Towards hardness, suppose $m : \mathbb{R}^n \to \mathbb{R}$ is a monomial with a negative coefficient, i.e.,
		\[
		m(x_1, \ldots, x_n) \ = \ - c \prod_{j=1}^{d} x_{i_j} \qquad \text{for some } c \gr 0 \text{ and } i_1, \ldots, i_d \in \{1, \ldots, n\}\,.
		\]
		Then we have
		\begin{align*}
			m(x_1, \ldots, x_n) \ &= \ - c \prod_{j=1}^{d} x_{i_j} \ = \  c (1-x_{i_1}) \prod_{j=2}^{d} x_{i_j} - c \prod_{j=2}^{d} x_{i_j} \ = \ \ldots \\
			&= \ -c + \sum_{k=1}^d c(1-x_{i_k}) \prod_{j=k+1}^d x_{i_j} \,.
		\end{align*}
		
		We reduce from the problem from \cref{lem:ETR-2}.
		Let $p : \mathbb{R}^n \to \mathbb{R}$ be a multivariate polynomial of degree at most~$6$.
		By rewriting each monomial of~$p$ that has a negative coefficient using the pattern above, we can write $p(x) = -\theta + q(x)$ for some $\theta \ge 0$ and some $q: \mathbb{R}^n \to \mathbb{R}$ of the form $q(x) = \sum_{j=1}^k f_j(x)$ where each $f_j(x)$ is a product of a nonnegative coefficient and at most $6$~terms of the form $x_i$ or $(1-x_i)$.
		As long as there is an $f_j(x_1, \ldots, x_n)$ of degree less than~$6$, we can replace it by the two summands $x_1 f_j(x_1, \ldots, x_n)$ and $(1-x_1) f_j(x_1, \ldots, x_n)$.
		So we can assume that every $f_j(x)$ has the required form.
		For all $x \in \mathbb{R}^n$ we have that $p(x) \gr 0$ if and only if $q(x) \gr \theta$, as required.
	\end{proof}

	\QT{
		To show that the memoryless distance minimisation problem is $\ETR$-hard, we reduce from the problem in \cref{lem:ETR-3}.
		We give a brief outline of the reduction.
		Given a multivariate polynomial $p: \mathbb{R}^n \to \mathbb{R}$ of the form as in \cref{lem:ETR-3},
		we construct an MDP with initial states $s_1$ and $s_2$ such that each assignment $x \in [0, 1]^n$ corresponds to a memoryless strategy $\alpha(x)$ of the MDP.
		The distance of $s_1$ and $s_2$ in the LMC induced by the memoryless strategy $\alpha(x)$ is $1 - c \cdot p(x)$ where $c$ is a constant.
		Therefore, there exists $x \in [0,1]^n$ with $p(x) \gr \theta$ if and only if there exists a memoryless strategy $\alpha(x)$ such that the distance of $s_1$ and $s_2$ is less than $1 - c \cdot \theta$.
	}
	
	\begin{theorem} \label{thm:memmin}
		The memoryless distance minimisation problem is $\ETR$-hard.
	\end{theorem}
	\begin{proof}
		We reduce from the problem from \cref{lem:ETR-3}.
		Let $\theta \ge 0$ and let $p : \mathbb{R}^n \to \mathbb{R}$ be a multivariate polynomial of the form $p(x) = \sum_{j=1}^m f_j(x)$ where each $f_j(x_1, \ldots, x_n)$ is a product of a nonnegative coefficient and $6$~terms of the form $x_i$ or $(1-x_i)$.
		Let us write $f_j(x_1, \ldots, x_n) = c_j \prod_{k=1}^6 x_{\ell(j,k)}$ where each $c_j \ge 0$ and each $\ell(j,k) \in \{-n, \ldots, -1, 1, \ldots, n\}$ and we use the notation $x_{-i}$ for $i \gr 0$ to mean $1-x_{i}$.
		We can assume that $\sum_{j=1}^m c_j = 1$ (otherwise, divide $\theta$ and each~$c_j$ by $\sum_{j=1}^m c_j$).
		
		\begin{figure}[t]	
			\centering

			\tikzstyle{BoxStyle} = [draw, circle, fill=black, scale=0.4,minimum width = 1pt, minimum height = 1pt]
			
			\begin{tikzpicture}[yscale=0.95,xscale=.6,>=latex',shorten >=1pt,node distance=3cm,on grid,auto]
				\node[state] (s1) at (-1,0) {$s_1$};
				\node[state] (u11) at (2,1) {$u_{11}$};
				\node[state] (v11) at (4,1) {$v_{11}$};
				\node[state] (w11) at (6,1) {$w_{11}$};
				\node[label] at (6,1.7) {\color{blue}$\mathbf{1}$};
				\node[state] (u12) at (8,1) {$u_{12}$};
				\node[state] (v12) at (10,1) {$v_{12}$};
				\node[state] (w12) at (12,1) {$w_{12}$};
				\node[label] at (12,1.7) {\color{blue}$\mathbf{1}$};
				\node[state] (u13) at (14,1) {$u_{13}$};
				\node[state] (v13) at (16,1) {$v_{13}$};
				\node[state] (w13) at (18,1) {$w_{13}$};
				\node[label] at (18,1.7) {\color{blue}$\mathbf{-2}$};
				\node[state] (u21) at (2,-1) {$u_{21}$};
				\node[state] (v21) at (4,-1) {$v_{21}$};
				\node[state] (w21) at (6,-1) {$w_{21}$};
				\node[label] at (6,-0.3) {\color{blue}$\mathbf{2}$};
				\node[state] (u22) at (8,-1) {$u_{22}$};
				\node[state] (v22) at (10,-1) {$v_{22}$};
				\node[state] (w22) at (12,-1) {$w_{22}$};
				\node[label] at (12,-0.3) {\color{blue}$\mathbf{4}$};
				\node[state] (u23) at (14,-1) {$u_{23}$};
				\node[state] (v23) at (16,-1) {$v_{23}$};
				\node[state] (w23) at (18,-1) {$w_{23}$};
				\node[label] at (18,-0.3) {\color{blue}$\mathbf{-4}$};
				\node[state] (up) at (21,0) {$u'$};
				\node[state] (x) at (21,-1.5) {$t$};

				\path[->] (s1) edge node [midway, above] {$\frac{1}{3}$} (u11);
				\path[->] (u11) edge node [midway, above] {} (v11);
				\path[->] (v11) edge node [midway, above] {} (w11);
				\path[->] (w11) edge node [midway, above] {} (u12);
				\path[->] (u12) edge node [midway, above] {} (v12);
				\path[->] (v12) edge node [midway, above] {} (w12);
				\path[->] (w12) edge node [midway, above] {} (u13);
				\path[->] (u13) edge node [midway, above] {} (v13);
				\path[->] (v13) edge node [midway, above] {} (w13);
				
				\path[->] (s1) edge node [midway, below] {$\frac{2}{3}$} (u21);
				\path[->] (u21) edge node [midway, above] {} (v21);
				\path[->] (v21) edge node [midway, above] {} (w21);
				\path[->] (w21) edge node [midway, above] {} (u22);
				\path[->] (u22) edge node [midway, above] {} (v22);
				\path[->] (v22) edge node [midway, above] {} (w22);
				\path[->] (w22) edge node [midway, above] {} (u23);
				\path[->] (u23) edge node [midway, above] {} (v23);
				\path[->] (v23) edge node [midway, above] {} (w23);
				\path[->] (w13) edge node [midway, above] {} (up);
				\path[->] (w23) edge node [midway, above] {} (up);
				\path[->] (up) edge node [midway, above] {} (x);
				\path (x) edge [loop below] node {} (x);
				
				\coordinate (du) at (-1.25,-4) {};
				
				\node[state] (s2) at (11.5,-2.5) {$s_2$};
				\node[state] (u) at (11.5,-4) {$u$};
				\node[state] (v1) at (1.5,-5.5) {$v_1$};
				\node[state] (v2) at (6.5,-5.5) {$v_2$};
				\node[state] (v3) at (11.5,-5.5) {$v_3$};
				\node[state] (v4) at (16.5,-5.5) {$v_4$};
				\node[state] (xp) at (21,-5.5) {$t'$};
				
				\node[BoxStyle] (dv11) at (0.25,-6.3) {};
				\node[BoxStyle] (dv12) at (2.75,-6.3) {};
				\node[BoxStyle] (dv21) at (5.25,-6.3) {};
				\node[BoxStyle] (dv22) at (7.75,-6.3) {};
				\node[BoxStyle] (dv31) at (10.25,-6.3) {};
				\node[BoxStyle] (dv32) at (12.75,-6.3) {};
				\node[BoxStyle] (dv41) at (15.25,-6.3) {};
				\node[BoxStyle] (dv42) at (17.75,-6.3) {};
				
				\node[state] (w1) at (0.25,-7.5) {$w_1$};
				\node[label] at (-.75,-7.5) {\color{blue}$\mathbf{1}$};
				\node[state] (w1m) at (2.75,-7.5) {$w_{-1}$};
				\node[label] at (1.5,-7.5) {\color{blue}$\mathbf{-1}$};
				\node[state] (w2) at (5.25,-7.5) {$w_{2}$};
				\node[label] at (4.25,-7.5) {\color{blue}$\mathbf{2}$};
				\node[state] (w2m) at (7.75,-7.5) {$w_{-2}$};
				\node[label] at (6.5,-7.5) {\color{blue}$\mathbf{-2}$};
				\node[state] (w3) at (10.25,-7.5) {$w_{3}$};
				\node[label] at (9.25,-7.5) {\color{blue}$\mathbf{3}$};
				\node[state] (w3m) at (12.75,-7.5) {$w_{-3}$};
				\node[label] at (11.5,-7.5) {\color{blue}$\mathbf{-3}$};
				\node[state] (w4) at (15.25,-7.5) {$w_{4}$};
				\node[label] at (14.25,-7.5) {\color{blue}$\mathbf{4}$};
				\node[state] (w4m) at (17.75,-7.5) {$w_{-4}$};
				\node[label] at (16.5,-7.5) {\color{blue}$\mathbf{-4}$};
				
				\path[->] (s2) edge node {} (u);
				\path[->] (u) edge node [midway, above] {$\frac{1}{5}$} (v1);
				\path[->] (u) edge node [below] {$\frac{1}{5}$} (v2);
				\path[->] (u) edge node [left] {$\frac{1}{5}$} (v3);
				\path[->] (u) edge node [midway, below] {$\frac{1}{5}$} (v4);
				\path[->] (u) edge node [midway, above] {$\frac{1}{5}$} (xp);
				\path (xp) edge [loop below] node {} (xp);
				
				\path[-] (v1) edge node {} (dv11);
				\path[-] (v1) edge node {} (dv12);
				\path[-] (v2) edge node {} (dv21);
				\path[-] (v2) edge node {} (dv22);
				\path[-] (v3) edge node {} (dv31);
				\path[-] (v3) edge node {} (dv32);
				\path[-] (v4) edge node {} (dv41);
				\path[-] (v4) edge node {} (dv42);
				
				\path[->] (dv11) edge node [midway, above] {} (w1);
				\path[->] (dv12) edge node [midway, above] {} (w1m);
				\path[->] (dv21) edge node [midway, above] {} (w2);
				\path[->] (dv22) edge node [midway, above] {} (w2m);
				\path[->] (dv31) edge node [midway, above] {} (w3);
				\path[->] (dv32) edge node [midway, above] {} (w3m);
				\path[->] (dv41) edge node [midway, above] {} (w4);
				\path[->] (dv42) edge node [midway, above] {} (w4m);
				
				\draw (node cs:name=w1,anchor=south) |- (17.75,-8.3);
				\draw (node cs:name=w1m,anchor=south) |- (17.75,-8.3);
				\draw (node cs:name=w2,anchor=south) |- (17.75,-8.3);
				\draw (node cs:name=w2m,anchor=south) |- (17.75,-8.3);
				\draw (node cs:name=w3,anchor=south) |- (17.75,-8.3);
				\draw (node cs:name=w3m,anchor=south) |- (17.75,-8.3);
				\draw (node cs:name=w4,anchor=south) |- (17.75,-8.3);
				\draw[->] (node cs:name=w4m,anchor=south) |- (17.75,-8.3) -| (du) -- (u); 
			\end{tikzpicture}
			
			\caption{
				An illustration of the proof of \cref{thm:memmin}.
				Consider the polynomial~$p$ with $p(x_1, x_2, x_3, x_4) = \frac13 x_1^2 (1-x_2) + \frac23 x_2 x_4 (1-x_4)$.
				This example polynomial has degree~$3$ (instead of degree~$6$ in the proof) to allow for a more succinct picture.
				The analogous construction from the reduction yields the shown MDP.
				\QT{The labels are written next to the states in blue, unlike the other figures in this paper where we usually use different colours to indicate different state labels. 
					We omit label~$0$.}
				There is a one-to-one correspondence between an assignment $x \in [0,1]^4$ and a memoryless strategy~$\alpha(x)$ in the MDP.
				It is such that \QT{$d(s_1,s_2) = 1 - \frac{p(x)}{5^4}$}, establishing a connection between an evaluation of~$p$ and the distance.
			}
			\label{fig:example}
		\end{figure}
		
		Construct an MDP \QT{which consists of two disjoint parts} as follows; see \cref{fig:example} for an illustration.
		\QT{The first part is an LMC.}
		Include states $u_{j,k}, v_{j,k}, w_{j,k}$ for each $j \in \{1, \ldots, m\}$ and each $k \in \{1, \ldots, 6\}$.
		Each $u_{j,k}, v_{j,k}$ has label~$0$, and each $w_{j,k}$ has label~$\ell(j,k)$.
		Each $u_{j,k}$ transitions with probability~$1$ to~$v_{j,k}$.
		Each $v_{j,k}$ transitions with probability~$1$ to~$w_{j,k}$.
		Each $w_{j,k}$, except those with $k=6$, transitions with probability~$1$ to~$u_{j,k+1}$.
		Include also states~$s_1$\QT{, $u'$ and $t$} with label~$0$.
		State~$s_1$ transitions with probability~$c_j$ to $u_{j,1}$, for each~$j$.
		\QT{
			State $u'$ transitions with probability~$1$ to~$t$.
			State $t$ is a sink state, that is, it transitions with probability~$1$ to itself.
			Also each $w_{j,6}$ transitions with probability~$1$ to~$u'$.
		}
		
		\QT{The second part is an MDP.}
		Include states \QT{$s_2, u, t'$} with label~$0$.
		State $s_2$ transitions with probability~$1$ to~$u$.
		Include also states $v_1, \ldots, v_n$, each with label~$0$.
		\QT{State $u$ transitions to each $v_i$ and $t'$ with probability~$\frac{1}{n+1}$.
			State $t'$ is a sink state.
		}
		Include also states $w_{-n}, \ldots, w_{-1}, w_1, \ldots, w_{n}$, where each $w_i$ has label~$i$.
		Each $v_i$ has two actions, one of which leads with probability~$1$ to~$w_i$, the other one with probability~$1$ to~$w_{-i}$.
		Each $w_i$ transitions with probability~$1$ to~$u$.
		
		Each assignment $x \in [0,1]^n$ corresponds to a memoryless strategy~$\alpha(x)$ such that in state~$v_i$ the memoryless strategy $\alpha(x)$ takes with probability~$x_i$ the action that leads to~$w_i$, and $\alpha(x)$~takes with probability~$1-x_i$ the action that leads to~$w_{-i}$.
		In fact, this mapping~$\alpha$ (from an assignment to a memoryless strategy) is a bijection.
		Fix an arbitrary $x \in [0,1]^n$ and consider the distances in the LMC induced by~$\alpha(x)$.
		For notational convenience, for any states $s,s'$ let us write $\overline{d}(s,s') := 1-d(s,s')$.
		Further, let us write $u_{j,7}$ for~\QT{$u'$}.
		
		Let $j \in \{1, \ldots, m\}$ and $k \in \{1, \ldots, 6\}$.
		Then we have
		\[
		\overline{d}(u_{j,k},u) \ = \ \frac{1}{n+1} \overline{d}(v_{j,k},v_{|\ell(j,k)|}) \ = \ \frac{1}{n+1} x_{\ell(j,k)} \overline{d}(w_{j,k}, w_{\ell(j,k)}) \ = \ \frac{1}{n+1} x_{\ell(j,k)} \overline{d}(u_{j,k+1}, u)\,.
		\]
		Since $\overline{d}(u_{j,7},u) = \overline{d}(u',u) = \frac{1}{n+1}$, it follows
		\[
		\overline{d}(u_{j,1},u) \ = \ \left(\frac{1}{n+1}\right)^7 \prod_{k=1}^6 x_{\ell(j,k)}\,.
		\]
		Hence,
		\begin{align*}
		\overline{d}(s_1,s_2) \ &= \ \sum_{j=1}^m c_j \overline{d}(u_{j,1},u) \ = \ \sum_{j=1}^m c_j \left(\frac{1}{n+1}\right)^7 \prod_{k=1}^6 x_{\ell(j,k)} \ = \ \left(\frac{1}{n+1}\right)^7 \sum_{j=1}^m f_j(x) \ \\
                                &= \ \frac{p(x)}{(n+1)^7}\,.
		\end{align*}
		Thus, we have $p(x) \gr \theta$ if and only if $\overline{d}(s_1,s_2) \gr \frac{\theta}{(n+1)^7}$ if and only if $d(s_1, s_2) \ls 1 - \frac{\theta}{(n+1)^7}$.
		This completes the hardness proof.
	\end{proof}
	
	\QT{The following theorem, proved in \cref{proof:thm:memmin-upper},  provides a matching upper bound. }
	
	\begin{restatable}{theorem}{theoremMemMinUpper}\label{thm:memmin-upper}
		The memoryless distance minimisation problem is in $\ETR$.
	\end{restatable}
	
	Together with \cref{thm:memmin} we obtain:
	\begin{corollary}
		The memoryless distance minimisation problem is $\ETR$-complete.
	\end{corollary}
	
	\section{General Strategies: Distance Minimisation}\label{section:general-strategies}
	In this section we consider the \emph{general distance minimisation problem} which, given an MDP, two states $s_1, s_2$ of the MDP, and a rational number $\theta$, asks whether there is a general strategy~$\alpha$ such that $d(s_1,s_2) \ls \theta$ holds in the LMC induced by~$\alpha$.
	
	To show that the general distance minimisation problem is undecidable, we establish a reduction from the emptiness problem for probabilistic automata.
	
	A probabilistic automaton is a tuple
	$\A = <Q, q_0, L, \delta, F>$ consisting of a finite set $Q$ of states, an initial state $q_0 \in Q$,
	a finite set $L$ of letters, a transition function $\delta: Q \times L \to \Dist(Q)$ assigning to every state and letter a distribution over states, and a set $F$ of final states.
	We also extend $\delta$ to words, by letting $\delta(q_0, \varepsilon) = \mathbf{1}_{q_0}$ and $\delta(q_0, \sigma w) = \sum_{q \in Q}\delta(q_0, \sigma)(q)\delta(q, w)$ for $\sigma \in L$ and $w \in L^{*}$.
	For a state $q \in Q$, $\A_q$ is the probabilistic automaton obtained from $\A$ by making $q$ the initial state.
	
	We write $\Pr_{\A}(w) = \sum_{q \in F} \delta(q_0, w)(q)$ to denote the probability that $\A$ accepts a word $w$.
	The emptiness problem asks, given a probabilistic automaton $\A$, whether there exists a word $w$ such that $\Pr_{\A}(w) \gr \frac{1}{2}$ holds.
	The probabilistic automaton $\A$ is called empty if no such word exists.
	\QT{This problem is known to be undecidable \cite{Fijalkow2017,Paz2014}, even for probabilistic automata with only two letters \cite{BC2003}\footnote{It is stated in \cite[Theorem 2.1]{BC2003} that the emptiness problem with unfixed threshold $\lambda$, i.e., whether there exists a word $w$ such that $\Pr_{\A}(w) \gr \lambda$, is undecidable for probabilistic automata with only two letters. It is easy to adapt the proof to show undecidability of the emptiness problem with fixed threshold~$\frac{1}{2}$.}.}
	
	Let $\A = <Q, q_0, L, \delta, F>$ be a probabilistic automaton; without loss of generality we
	assume that $q_0 \not\in F$ and $L = \{a, b\}$.
	We construct an MDP $\D$ with states $s_1$ and $s_2$ and a number $\theta$ such that $\A$ is nonempty if and only if there is a general strategy such that $d(s_1,s_2) \ls \theta$ in the induced LMC.
	
	\QT{
		Let us first outline the idea of the construction.
		Our MDP includes the part shown in \cref{fig:general-MDP1}, where after a random word $w \in L^*$ is produced, the strategy must choose between taking the transition to $x$ or to~$y$.
		\Cref{lem:general-distance-expression} below characterises the distance of $s_1$ and $s_2$ under strategy $\alpha$ in terms of $\alpha$ and $\Pr_{\A}$.
		It follows from \cref{lem:general-distance-expression} that the following strategy minimises the distance: if the random word $w$ satisfies $\Pr_{\A}(w) \le \frac{1}{2}$, choose the transition to~$x$; otherwise choose the transition to~$y$.
	}
	\SK{
		Setting $\theta$ as the distance under the strategy that always chooses the transition to~$x$, we obtain that the distance can be made less than~$\theta$ if and only if there is a word~$w$ with $\Pr_{\A}(w) \gr \frac{1}{2}$.
	}
	
	\begin{figure}[!htb]	
		\centering
		\tikzstyle{BoxStyle} = [draw, circle, fill=black, scale=0.4,minimum width = 1pt, minimum height = 1pt]
		
		\begin{tikzpicture}[xscale=.6,>=latex',shorten >=1pt,node distance=3cm,on grid,auto]
			\node[state, fill=orange!20] (a) at (-2, -2.5) {$a$};
			\node[state, fill=blue!20] (b) at (2, -2.5) {$b$};
			\node[state] (dollar) at (0,-4.2){$\$$};
			\node[BoxStyle] (mx) at (-2, -4.5) {};
			\node[BoxStyle] (my) at (2, -4.5) {};
			\node[label] at (-2.7, -4.5) {$\m_x$};
			\node[label] at (2.7, -4.5) {$\m_y$};
			\node[state, fill=red!20] (x) at (-2, -5.7) {$x$};
			\node[state, fill=OliveGreen!30] (y) at (2, -5.7) {$y$};
			
			

			
			\path (a) edge [loop left] node {$\frac{1}{3}$} (a);
			\path (b) edge [loop right] node {$\frac{1}{3}$} (b);
			\path[->] (a) edge [bend right=15] node [above] {$\frac{1}{3}$} (b);
			\path[->] (b) edge [bend right=15] node [above] {$\frac{1}{3}$} (a); 								
			\path[->] (a) edge node [left] {$\frac{1}{3}$} (dollar);
			\path[->] (b) edge node [right] {$\frac{1}{3}$} (dollar); 						
			\path[-] (dollar) edge node {} (mx);
			\path[-] (dollar) edge node {} (my);
			\path[->] (mx) edge node [pos=0.6, left] {$1$} (x);
			\path[->] (my) edge node [pos=0.6, right] {$1$} (y);
			\path (x) edge [loop left] node {$1$} (x);
			\path (y) edge [loop right] node {$1$} (y);
		\end{tikzpicture}
		
		\caption{
			The first part of the MDP $\D$.
			The $\$$ state is the only one that has nondeterministic choices: it has two available actions, $\m_x$ and $\m_y$.
			The default action $\m$ for the other states is omitted. 
			Different colours indicate different state labels.
		}
		\label{fig:general-MDP1}
	\end{figure}
	
	\SK{		We now give the details of the construction.}
	The MDP $\D = <S, \Act, L', \varphi, \ell>$ consists of two disjoint parts as follows; see \cref{fig:general-MDP1} and \cref{fig:general-MDP2}.
	The set of actions is $\Act = \{\m, \m_x, \m_y\}$.
	The set of labels is $L' = \{a, b, \$, x, y\}$.
	
	The first part is an MDP shown in \cref{fig:general-MDP1}.
	Its set of states is $\{a, b, \$, x, y\}$.
	The state $s_1$ is defined to be $a$.
	The transitions $\varphi$ are defined as follows:
	\begin{itemize}
		\item
		The state $a$ (resp. $b$) transitions with uniform probability to its three successors $a$, $b$ and $\$$, that is, $\varphi(s, \m)(a) = \varphi(s, \m)(b) = \varphi(s, \m)(\$) = \frac{1}{3}$ for $s \in \{a, b\}$.
		\item
		The state $\$$ has two actions $\m_x$ and $\m_y$; the action $\m_x$ goes with probability $1$ to $x$ and the action  $\m_y$ goes with probability $1$ to $y$. That is,
		$\varphi(\$, \m_x)(x) = \varphi(\$, \m_y)(y) = 1$.
		\item
		The states $x$ and $y$ are sink states, that is, $\varphi(s, \m)(s) = 1$ for $s \in \{x, y\}$.
	\end{itemize}
	Each of the states is labelled with its name, that is, $\ell(s) = s$ for $s \in \{a, b, \$, x, y\}$.
	This sub-MDP ``is almost'' an MC, in the sense that a strategy $\alpha$ does not influence its behaviour until eventually a transition to $x$ or $y$ is taken.
	Since $a$, $b$, $x$ and $y$ have only one available action, we may omit the default action $\m$ in the paths that contain $\m$ only.
	For example, we may write $s_1ab\$$ to represent the path $s_1 \m a \m b \m \$$.
	
	\begin{figure}[!htb]	
		\centering
		\tikzstyle{BoxStyle} = [draw, circle, fill=black, scale=0.4,minimum width = 1pt, minimum height = 1pt]
		
		\begin{tikzpicture}[xscale=.6,>=latex',shorten >=1pt,node distance=3cm,on grid,auto]
			\node[state] (q) at (0, 0){$q$};
			\node[state] (q1) at (5, 1) {$q_1$};
			\node[state] (q2) at (5, -1) {$q_2$};
			\node[label] at (0, -0.8) {$q \not\in F$};
			
			\node[label] at (2,-3) {\parbox{3.5cm}{$p_1$ and $p_2$ are transition probabilities: \\$p_1 = \delta(q, a)(q_1)$ and $p_2 = \delta(q, b)(q_2)$.}};
			
			\path[->] (q) edge node [midway, above] {$p_1, a$} (q1);
			\path[->] (q) edge node [midway, below] {$p_2, b$} (q2);
			
			\node[state, fill=orange!20] (aq) at (10, 0.2) {$(a, q)$};
			\node[state, fill=orange!20] (aq1) at (16, 1.2) {$(a, q_1)$};
			\node[state, fill=blue!20] (bq2) at (16, -0.8) {$(b, q_2)$};
			\node[state] (dollarx) at (10, -1.5) {$\$_x$};
			\node[state, fill=red!20] (x) at (10, -3) {$x'$};
			
			\path[->] (aq) edge node [midway, above] {$\frac{p_1}{3}$} (aq1);
			\path[->] (aq) edge node [midway, below] {$\frac{p_2}{3}$} (bq2);
			\path[->] (aq) edge node [midway, right] {$\frac{1}{3}$} (dollarx);
			\path[->] (dollarx) edge node [midway, right] {$1$} (x);
			\path (x) edge [loop right] node {$1$} (x);
			
			\node[state,accepting] (f) at (0, -6){$f$};
			\node[state] (q3) at (5, -5) {$q_3$};
			\node[state] (q4) at (5, -7) {$q_4$};
			\node[label] at (0, -6.8) {$f \in F$};
			
			\node[label] at (2,-9) {\parbox{3.5cm}{$p_3$ and $p_4$ are transition probabilities: \\$p_3 = \delta(f, a)(q_3)$ and $p_4 = \delta(f, b)(q_4)$.}};
			
			\path[->] (f) edge node [midway, above] {$p_3, a$} (q3);
			\path[->] (f) edge node [midway, below] {$p_4, b$} (q4);
			
			\node[state, fill=orange!20] (af) at (10, -5.8) {$(a, f)$};
			\node[state, fill=orange!20] (aq3) at (16, -4.8) {$(a, q_3)$};
			\node[state, fill=blue!20] (bq4) at (16, -6.8) {$(b, q_4)$};
			\node[state] (dollary) at (10, -7.5) {$\$_y$};
			\node[state, fill=OliveGreen!30] (y) at (10, -9) {$y'$};
			
			\path[->] (af) edge node [midway, above] {$\frac{p_3}{3}$} (aq3);
			\path[->] (af) edge node [midway, below] {$\frac{p_4}{3}$} (bq4);
			\path[->] (af) edge node [midway, right] {$\frac{1}{3}$} (dollary);
			\path[->] (dollary) edge node [midway, right] {$1$} (y);
			\path (y) edge [loop right] node {$1$} (y);
			
			\node[label] at (2.5, 2.3) {The probabilistic automaton $\A$};
			\node[rectangle,draw,dashed, minimum width = 5cm, minimum height = 12cm] at (2.5,-4.1) {};
			\node[label] at (13, 2.3) {The second part of $\D$};
			\node[rectangle,draw,dashed, minimum width = 6cm, minimum height = 12cm] at (13,-4.1) {};
		\end{tikzpicture}
		
		\caption{
			The second part of the MDP $\D$ is an LMC, constructed from the probabilistic automaton $\A$.
			The default deterministic action $\m$ for all states is omitted.
			The state $(b, q)$ in the MDP $\D$, where $q \in Q$, has the same transitions as the state $(a, q)$;
			it is labelled with $b$.
		}
		\label{fig:general-MDP2}
	\end{figure}
	
	The other part of $\D$ is an LMC constructed from $\A$ as follows; see \cref{fig:general-MDP2}.
	The set of states is $(L \times Q) \cup \{\$_x, \$_y, x', y'\}$.
	The state $s_2$ is defined to be $(a, q_0)$.
	
	We describe the transitions of the LMC using the transition function $\delta$ of $\A$.
	Consider a letter $\sigma \in L$ and a state $q \in Q$.
	The state $(\sigma, q)$ with probability $\frac{1}{3}$ simulates the probabilistic automaton $\A$ reading the letter $a$, and with probability $\frac{1}{3}$ simulates the probabilistic automaton $\A$ reading the letter $b$.
	That is, $\varphi\big((\sigma, q), \m \big) \big( (a, q') \big) = \frac{1}{3}\delta(q, a)(q')$ and $\varphi\big((\sigma, q), \m \big) \big( (b, q') \big) = \frac{1}{3}\delta(q, b)(q')$.
	
	For the remaining probability of $\frac{1}{3}$, we distinguish the following two cases:
	\begin{itemize}
		\item
		If $q \not\in F$, the state $(\sigma, q)$ transitions to $\$_x$ with probability $\frac{1}{3}$, that is, $\varphi\big((\sigma, q), \m \big) (\$_x) = \frac{1}{3}$.
		\item
		Otherwise, if $q \in F$, the state $(\sigma, q)$ transitions to $\$_y$ with probability $\frac{1}{3}$, that is,  $\varphi\big((\sigma, q), \m \big) (\$_y) = \frac{1}{3}$.
	\end{itemize}
	The state $\$_x$ (resp. $\$_y$) transitions with probability one to the sink state $x'$ (resp. $y'$).
	That is, $\varphi(\$_x, \m)(x') = \varphi(\$_y, \m)(y') = \varphi(x', \m)(x') = \varphi(y', \m)(y') = 1$.
	
	A state $(\sigma, q) \in L \times Q$ is labelled with $\sigma$.
	The states $\$_x$ and $\$_y$ are labelled with $\$$.
	The states $x'$ and $y'$ are labelled with $x$ and $y$, respectively.
	
	Given a general strategy $\alpha$, the next lemma expresses the distance between $s_1$ and $s_2$ in terms of $\alpha$ and $\Pr_{\A}$.
	The proof is technical and can be found in \cref{proof:lem:general-distance-expression}.
	\begin{lemma}\label{lem:general-distance-expression}
		For any general strategy $\alpha$, we have
		\[d_\alpha(s_1, s_2)
		= \sum_{w \in L^*} \frac{1}{3^{|w| + 1}} \big( (1-\textstyle\Pr_{\A}(w)) \alpha(s_1 w
		\$)(\m_y) + \Pr_{\A}(w) \alpha(s_1 w \$)(\m_x) \big).\]
	\end{lemma}
	
	Using \cref{lem:general-distance-expression}, we prove the main theorem of this section:
	
	\begin{theorem}\label{thm:general-undecidable}
		The general distance minimisation problem is undecidable.
	\end{theorem}
	
	\begin{proof}
		We reduce from the emptiness problem for probabilistic automata.
		Let $\A = <Q, q_0, L, \delta, F>$ be a probabilistic automaton; without loss of generality we assume that $q_0 \not\in F$ and $L = \{a, b\}$.
		Let $\D$ be the MDP constructed from $\A$ shown in \cref{fig:general-MDP1,fig:general-MDP2}.
		
		Let $\alpha_x$ be the memoryless strategy that chooses the action $\m_x$ whenever it is in state $\$$, that is, $\alpha_x(s_1w\$) = \mathbf{1}_{\m_x}$ for all $w \in L^{*}$.
		Let $\theta$ be the distance between $s_1$ and $s_2$ in the LMC $\D(\alpha_x)$.
		It can be computed in polynomial time \cite{CvBW2012}.
		We show in \cref{proof:thm:general-undecidable} that there is a word $w \in L^{*}$ such that $\Pr_{\A}(w) \gr \frac{1}{2}$ ($\A$ is nonempty) if and only if there is a general strategy $\alpha$ such that $d_{\alpha}(s_1,s_2) \ls \theta$ in the induced LMC.
	\end{proof}
	
	\section{General Strategies: Distance Less Than One}\label{section:distance-less-than-one}
	In this section, we consider the distance less than one problem which, given an MDP and two states, asks whether there is a general strategy such that the two states have probabilistic bisimilarity distance less than one in the LMC induced by the general strategy.
	The challenge here is that general strategies induce, in general, LMCs with infinitely many states.
	
	We show that the distance less than one problem is EXPTIME-complete.
	We prove the upper and lower bound in \cref{subsection:EXPTIME-ub,subsection:EXPTIME-hardness}, respectively.
	
	\subsection{Membership in EXPTIME}\label{subsection:EXPTIME-ub}
	Let $\M  = <S, L, \tau, \ell>$ be \QT{a (possibly infinite)} LMC. We partition the set $S^2$ of state pairs into
	\[
	\begin{array}{rcl}
		S^2_0 & = & \{\, (s, t) \in S^2 \mid s \sim t \,\}\\
		S^2_1 & = & \{\, (s, t) \in S^2 \mid \ell(s) \not= \ell(t) \,\}\\
		S^2_? & = & S^2 \setminus (S^2_0 \cup S^2_1)\,.
	\end{array}
	\]
	
	We call $T : S^2_? \to \Dist(S^2)$ a \emph{policy} for the LMC if for all $(s, t) \in S^2_?$ we have $T(s, t) \in \Omega(\tau(s), \tau(t))$.
	We write $\mathcal{T}$ for the set of policies.
		Given a policy $T \in \mathcal{T}$,
		the Markov chain $\C_{\M}^{T} = <S^2, \tau'>$ induced by $T$ is defined by
		\[
		\begin{array}{rcll}
			\tau'\big((u, v)\big)\big((u, v)\big) &=& 1 &\mbox{if $(u, v) \in S_0^2 \cup S_1^2$;}\\
			\tau'\big((u, v)\big)\big((x, y)\big) &=& T(u, v)(x, y) & \mbox{otherwise.} 
		\end{array}
		\]	
	For $(s,t) \in S^2$ and a set of state pairs $Z \subseteq S^2$ we write $\R_{\M}^T((s,t),Z) \in [0,1]$ for the probability that in the Markov chain~$\C_{\M}^T$ the state $(s,t)$ reaches a state $(u,v) \in Z$.
	
	By \cite[Theorem~4, Proposition~5]{TvB2018}, the following proposition holds.
	
	\begin{proposition}\label{proposition:lmc-distance-graph-reachability2}
		Let $\M= <S, L, \tau, \ell>$ be a \emph{finite} LMC and $s, t \in S$.
		We have $d(s,t) \ls 1$ if and only if there exists a policy $T$ such that 
		$\R_\M^T((s,t),S_0^2) \gr 0$.
	\end{proposition}
	
	The ``only if'' direction of \cref{proposition:lmc-distance-graph-reachability2} does not generally hold for LMCs with infinite state space, as the following example shows.
	\begin{figure}[t]
		\begin{subfigure}{1\textwidth}
			\centering
			\tikzstyle{BoxStyle} = [draw, circle, fill=black, scale=0.4,minimum width = 1pt, minimum height = 1pt]
			\begin{tikzpicture}[xscale=.6,>=latex',shorten >=1pt,node distance=3cm,on grid,auto]
				\node[state,fill=red!20] (s-1) at (-3,0) {$s_{0}$};
				\node[state] (s0) at (0,0) {$s_1$};
				\node[state] (s1) at (3,0) {$s_2$};
				\node[state] (s2) at (6,0) {$s_3$};
				\node[label] at (7.5,0) {$\cdots$};
				
				\node[state] (t0) at (12,0) {$t$};
				
				\path (s-1) edge [loop left] node {$1$} (s-1);
				\path[->] (s0) edge [bend left=15] node {$\frac{1}{3}$} (s-1);
				\path[->] (s0) edge [bend left=15] node {$\frac{2}{3}$} (s1);
				\path[->] (s1) edge [bend left=15] node {$\frac{2}{3}$} (s2);
				\path[->] (s1) edge [bend left=15] node {$\frac{1}{3}$} (s0);
				\path[->] (s2) edge [bend left=15] node {$\frac{1}{3}$} (s1);
				
				\path (t0) edge [loop right] node {$1$} (t0);
			\end{tikzpicture}
			\caption{An infinite LMC $\M$.}
			\label{fig:example-infinite-lmc}
		\end{subfigure}
		~
		\begin{subfigure}{1\textwidth}
			\centering
			\tikzstyle{BoxStyle} = [draw, circle, fill=black, scale=0.4,minimum width = 1pt, minimum height = 1pt]
				\begin{tikzpicture}[xscale=.5,>=latex',shorten >=1pt,node distance=3cm,on grid,auto]
					\node[state,inner sep=1pt] (s0t0) at (-4,0) {$(s_0,t)$};
					\node[state,inner sep=1pt] (s1t0) at (0,0) {$(s_1,t)$};
					\node[state,inner sep=1pt] (s2t0) at (4,0) {$(s_2, t)$};
					\node[state,inner sep=1pt] (s3t0) at (8,0) {$(s_3, t)$};
					\node[label] at (10,0) {$\cdots$};
					
					\path (s0t0) edge [loop left] node {$1$} (s0t0);
					\path[->] (s1t0) edge [bend left=10] node {$\frac{1}{3}$} (s0t0);
					\path[->] (s1t0) edge [bend left=10] node {$\frac{2}{3}$} (s2t0);
					\path[->] (s2t0) edge [bend left=10] node {$\frac{2}{3}$} (s3t0);
					\path[->] (s2t0) edge [bend left=10] node {$\frac{1}{3}$} (s1t0);
					\path[->] (s3t0) edge [bend left=10] node {$\frac{1}{3}$} (s2t0);
				\end{tikzpicture}
			\caption{The Markov chain $\C_{\M}^T$.}
			\label{fig:example-graph-lmc}
		\end{subfigure}
		\caption{
			(a) An infinite state LMC $\M$ with an infinite state space $S = \big\{s_i \mid i \in \{0, 1, 2, \ldots \}\big\} \cup \{t\}$.
			All states have the same label except $s_{0}$.
			The states $s_{0}$ and $t$ are sink states, that is, $\tau(s_{0})(s_{0}) = \tau(t)(t) =  1$.
			Each $s_{i}$ where $i \in \{1, 2, \ldots\}$ transitions to $s_{i-1}$ with probability $\frac{1}{3}$ and $s_{i+1}$ with probability $\frac{2}{3}$.
			(b) The Markov chain $\C_{\M}^T$ induced by an arbitrary policy $T$ in which only the states reachable from $(s_1, t)$ are shown.
			The shown part of $\C_{\M}^T$ is the same for every policy $T$.
		}
		\label{fig:infinite-LMC-counterexample}
	\end{figure}
	\begin{example} \label{example:drift-1}
		Consider the LMC $\M$ in \cref{fig:example-infinite-lmc}.
		Let $T$ be an arbitrary policy for $\M$.
		We have $T(s_i, t)(s_{i-1}, t) = \frac{1}{3}$ and $T(s_i, t)(s_{i+1}, t) = \frac{2}{3}$ for all $i \in \{1, 2, \ldots\}$.
		The Markov chain $\C_{\M}^T$ induced by $T$ is shown in \cref{fig:example-graph-lmc};
		we only show the states that are reachable from $(s_1, t)$.
		The shown part of $\C_{\M}^T$ is the same for every policy.
		
		We have $d(s_i,t) = \frac{1}{2^{i}}$ for all $i \in \{0, 1, 2, \ldots\}$.
		In the Markov chain $\C_{\M}^T$, all state pairs that $(s_1, t)$ can reach have distances greater than zero:
		for all $i \in \{1, 2, \ldots\}$ the pair $(s_1, t)$ can reach $(s_i,t)$ and we have $d(s_i,t) = \frac{1}{2^{i}} \gr 0$.
		\lipicsEnd
	\end{example}
	
	The following theorem follows from \cite[Theorem 6.1.7]{Tang2018} for LMCs with finite state space.
	The same proof, see \cref{proof:theorem:reach}, works for LMCs with infinite state space.
	\begin{theorem}
		\label{theorem:reach}
		Let $\M= <S, L, \tau, \ell>$ be an LMC.
		There is a policy $T \in \mathcal{T}$ such that we have
		\[
		d(s,t) \ = \ \R_\M^T((s,t),S_1^2) \ \le \ \R_\M^{T'}((s,t),S_1^2) \quad \text{for all $(s,t) \in S^2$ and all $T' \in \mathcal{T}$.}
		\]
		In short, $\displaystyle d = \min_{T \in \mathcal{T}} \R_\M^T(\cdot,S_1^2)$.
	\end{theorem}
	
	The following corollary of \cref{theorem:reach} is similar to \cref{proposition:lmc-distance-graph-reachability2} but holds even for infinite-state LMCs.%
	\begin{corollary}\label{corollary:lmc-distance-graph-reachability}
		Let $\M= <S, L, \tau, \ell>$ be an LMC and $s, t \in S$.
		We have $d(s,t) \ls 1$ if and only if there exists a policy $T$ such that $\R_\M^T((s,t),S_1^2) \ls 1$.
		In particular, if there is a policy~$T$ with $\R_\M^T((s,t),S_0^2) \gr 0$ then $d(s,t) \ls 1$.
	\end{corollary}
	
	\Cref{corollary:lmc-distance-graph-reachability} falls short of an ``if and only if'' connection between distance less than one and bisimilarity.
	Indeed, as we have seen, \cref{proposition:lmc-distance-graph-reachability2} does not always hold in infinite-state LMCs.
	However, the key technical insight of this section is that a version of \cref{proposition:lmc-distance-graph-reachability2} holds for (finite-state) MDPs and general strategies.
	More precisely, the following proposition characterises the existence of a strategy such that the distance is less than one.
	\begin{proposition}\label{prop:mdp-distance-ls1}
		Let $\D = <S, \Act, L, \varphi, \ell>$ be an MDP, and let $s, t \in S$.
		There exists a strategy~$\alpha''$ with $d_{\D(\alpha'')}(s, t) \ls 1$ if and only if
		there are strategies $\alpha, \alpha'$, a policy~$T$ for the LMC~$\D(\alpha)$, two states $u, v \in S$ and two paths $\rho_1,\rho_2 \in \Paths(\D)$ with $u = \last(\rho_1)$ and $v = \last(\rho_2)$, such that $\R_{\D(\alpha)}^T((s,t),\{(\rho_1,\rho_2)\}) \gr 0$ and $u$ and $v$ are probabilistically bisimilar in the LMC~$\D(\alpha')$.
	\end{proposition}
	The more difficult direction of the proof is the ``only if'' direction.
	It is based on L\'evy's zero-one law, several applications of the Bolzano-Weierstrass theorem, and a characterisation of probabilistic bisimilarity in MDPs in terms of an ``attacker-defender'' game defined \cite[Section~3.1]{KT2022}.
	
	The starting point of the proof of \cref{prop:mdp-distance-ls1} is the following statement, which follows from \cref{theorem:reach,corollary:lmc-distance-graph-reachability} using L\'evy's zero-one law.%
	\begin{corollary}\label{corollary:smaller-and-smaller}
		Let $\M= <S, L, \tau, \ell>$ be an LMC and $s, t \in S$ with $d(s,t) \ls 1$.
		There exists a policy $T$ such that for all $\varepsilon \gr 0$ there is $(u,v) \in S^2$ with $d(u,v) \le \varepsilon$ and $\R_\M^T((s,t),\{(u,v)\}) \gr 0$.
	\end{corollary}
	\begin{proof}
		Let $\M= <S, L, \tau, \ell>$ be an LMC and $s, t \in S$ with $d(s,t) \ls 1$.
		By \cref{corollary:lmc-distance-graph-reachability} there exists a policy $T$ such that $\R_\M^T((s,t),S_1^2) \ls 1$.
		By L\'evy's zero-one law, the probability in~$\C_\M^T$ is one that a random run $(s_0,t_0) (s_1,t_1) \ldots$ started from $(s_0,t_0) = (s,t)$ satisfies one of the following conditions:
		\begin{enumerate}
			\item the sequence $\R_\M^T((s_0,t_0),S_1^2), \R_\M^T((s_1,t_1),S_1^2), \ldots$ converges to~$1$ and $S_1^2$ is reached;
			\item the sequence $\R_\M^T((s_0,t_0),S_1^2), \R_\M^T((s_1,t_1),S_1^2), \ldots$ converges to~$0$ and $S_1^2$ is not reached.
		\end{enumerate}
		Event~1 can be equivalently characterised by saying that $S_1^2$ is reached.
		Since $\R_\M^T((s,t),S_1^2) \ls 1$, Event~2 happens with a positive probability.
		It follows that in~$\C_\M^T$ there exists a run $(s_0,t_0) (s_1,t_1) \ldots$ started from $(s_0,t_0) = (s,t)$ such that $\R_\M^T((s_0,t_0),S_1^2), \R_\M^T((s_1,t_1),S_1^2), \ldots$ converges to~$0$.
		Let $\varepsilon \gr 0$.
		Then there exists $(u,v) \in S^2$ such that $\R_\M^T((u,v),S_1^2) \le \varepsilon$ and $\R_\M^T((s,t),\{(u,v)\}) \gr 0$.
		By \cref{theorem:reach} it follows that $d(u,v) \le \varepsilon$.
	\end{proof}
	\begin{example} \label{example:drift-2}
		Consider again \cref{example:drift-1}.
		We have $d(s_1,t) = \frac12$.
		\Cref{corollary:smaller-and-smaller} asserts that there is a policy~$T$ such that for all $\varepsilon \gr 0$, in~$\C_\M^T$ the pair $(s_1,t)$ can reach $(u,v) \in S^2$ with $d(u,v) \le \varepsilon$.
		Indeed, take an arbitrary policy~$T$.
		Given any $\varepsilon \gr 0$ choose $i$ with $\frac{1}{2^i} \le \varepsilon$.
		Then $(s_1,t)$ can reach $(s_i,t)$ and $d(s_i,t) = \frac{1}{2^i} \le \varepsilon$.
		\lipicsEnd
	\end{example}
	
	See \cref{proof:prop:mdp-distance-ls1} for the rest of the proof of \cref{prop:mdp-distance-ls1}.
	\Cref{prop:mdp-distance-ls1} is the key to proving the following result.
	
	\begin{theorem}\label{thm:distance-ls-one-exptime}
		The distance less than one problem is in EXPTIME.
	\end{theorem}
	\begin{proof}
		Let $<S, \Act, L, \varphi, \ell>$ be an MDP.
		Abusing the notation from the beginning of \cref{subsection:EXPTIME-ub}, let us define
		\[
		\begin{array}{rcl}
			S^2_0 & = & \{\, (s, t) \in S^2 \mid \exists\; \alpha' \text{ such that } s, t \text{ are probabilistically bisimilar in } \D(\alpha') \,\}\\
			S^2_1 & = & \{\, (s, t) \in S^2 \mid \ell(s) \not= \ell(t) \,\}\\
			S^2_? & = & S^2 \setminus (S^2_0 \cup S^2_1)\,.
		\end{array}
		\]
		By \cite[Theorem~7]{KT2022} the set $S^2_0$ can be computed in exponential time.
		Consider the elements of~$S^2$ as vertices of a directed graph with set of edges
		\begin{align*}
			E \ &:= \ \{(z,z) \mid z \in S^2_0 \cup S^2_1\} \ \cup \\
			&\left\{\big((s_1,s_2),(t_1,t_2)\big) \in S^2_? \times S^2 \mid \forall\; i \in \{1,2\}\ \exists\;\m_i \in \Act(s_i) : \support(\varphi(s_i,\m_i)) \ni t_i\right\}\,.
		\end{align*}
		After $S^2_0$ has been computed (in exponential time), the directed graph $G := (S^2, E)$ can be computed in polynomial time, and given two states $s, t \in S$, it can be checked in polynomial time if $S^2_0$ can be reached from~$(s,t)$ in~$G$.
		It follows from \cref{prop:mdp-distance-ls1} that this is the case if and only if there exists a strategy~$\alpha''$ with $d_{\D(\alpha'')}(s, t) \ls 1$.
	\end{proof}
	
	\subsection{EXPTIME-Hardness}\label{subsection:EXPTIME-hardness}
	\SK{Given an MDP and two (initial) states, the \emph{bisimilarity problem} asks whether there is a general strategy such that the two states are probabilistically bisimilar in the induced LMC.
		The bisimilarity problem was shown EXPTIME-complete in \cite[Theorem~7]{KT2022}.
		We show in \cref{proof:thm:distance-ls-one-exptimehard} that it can be reduced to the distance less than one problem.
		This gives us the following theorem.
	}
	\begin{theorem}\label{thm:distance-ls-one-exptimehard}
		The distance less than one problem is EXPTIME-hard.
	\end{theorem}
	Together with \cref{thm:distance-ls-one-exptime} we obtain:
	\begin{corollary}
		The distance less than one problem is EXPTIME-complete.
	\end{corollary}
	
	\section{Conclusion}\label{section:conclusion}
	Motivated by probabilistic noninterference, a security notion, we have settled the decidability and complexity of the most natural bisimilarity distance minimisation problems of MDPs under memoryless and general strategies.
	
	Specifically, we have proved that the distance minimisation problem for memoryless strategies is $\ETR$-complete (which implies, in particular, that it is NP-hard and in PSPACE).
	In contrast, we have shown that the distance minimisation problem for general strategies is undecidable, reducing from the emptiness problem for probabilistic automata.
	
	We have also shown that it is EXPTIME-complete to decide if there are general strategies to make the probabilistic bisimilarity distance less than one.
	This extends a result from~\cite{KT2022} that the bisimilarity equivalence problem under general strategies is EXPTIME-complete.
	The key technical link we need here is natural but nontrivial to establish under general strategies: if there are general strategies such that two states have distance less than one, these two states can reach another pair of states which can be made probabilistic bisimilar.
	
	\SK{
		Distance \emph{maximisation} problems also relate to probabilistic noninterference, but in terms of antagonistic schedulers wanting to maximise the information leakage.
		The decidability and complexity of several distance maximisation problems in MDPs is still open, including the distance equals one problem for general strategies.
	}
	
	\bibliography{paper}
	
	\newpage\appendix\label{section:appendix}
	\section{Missing Proofs}

	\subsection{Proof of \texorpdfstring{\cref{thm:memmin-upper}}{Theorem \ref{thm:memmin-upper}}} \label{proof:thm:memmin-upper}
	Let $\D = <S, \Act, L, \varphi, \ell>$ be an MDP.
	Let $s_1, s_2$ be two states of $\D$ and $\theta$ be a rational number.

	We use numbers $x_{s, \m} \in [0,1]$ where $s \in S$ and $\m \in \Act(s)$ to characterise a memoryless strategy $\alpha$ for $\D$, that is, $\alpha(s)(\m) = x_{s, \m}$.
	Let the LMC induced by $\alpha$ be $\D(\alpha) = <S, L, \tau, \ell>$, where $\tau(s)(t) = \sum_{\m \in \Act(s)}\alpha(s)(\m) \cdot \varphi(s, \m)(t)$ for all $s, t \in S$.
	Numbers $w_{s,t,u,v}$ where $s,t,u,v \in S$ and $\ell(s) = \ell(t)$ characterise a coupling $\omega \in \Dist(S \times S)$ of $\tau(s)$ and $\tau(t)$.
	Numbers $d_{s, t} \in [0, 1]$ where $s, t \in S$ characterise a pseudometric.
	
	To decide if there exists a memoryless strategy such that $d(s_1, s_2) \ls \theta$, we check the validity of the following closed formula in the existential theory of the reals:
	\ $\exists x_{s, \m} \in [0,1]$ for all $s \in S$ and $\m \in \Act(s)$, $w_{s,t,u,v} \in [0,1]$ for all $s,t \in S$ such that $\ell(s) = \ell(t)$ and all $u, v \in S$ and $d_{s, t} \in [0,1]$ for all $s,t \in S$ such that
	\begin{enumerate}
		\item[-]
		$\sum_{\m \in \Act(s)}x_{s, \m} = 1$ for all $s \in S$;
		\item[-]
		$\sum_{v \in S}w_{s, t, u, v} = \tau(s)(u)$ for all $s, t \in S$ such that $\ell(s) = \ell(t)$ and all $u \in S$; 
		\item[-]
		$\sum_{u \in S}w_{s, t, u, v} = \tau(t)(v)$
		for all $s, t \in S$ such that $\ell(s) = \ell(t)$ and all $v \in S$;
		\item[-]
		$\sum_{u,v \in S}w_{s, t, u, v} d_{u,v}  \QT{=} d_{s, t}$ for all $s,t \in S$ such that $\ell(s) = \ell(t)$;
		\item[-]
		$d_{s,t} = 1$ for all $s,t \in S$ such that $\ell(s) \neq \ell(t)$;
		\item[-]
		$d_{s_1,s_2} \ls \theta$.
	\end{enumerate}
	We show in the following that this formula is valid if and only if the memoryless distance minimisation problem has answer ``yes''. \qedhere
	
	
	
	\begin{itemize}
		\item ``$\implies$'':
		Assume that the formula is valid.
		We define a memoryless strategy $\alpha$ for $\D$ by setting $\alpha(s)(\m) = x_{s, \m}$ for all $s \in S$ and $\m \in \Act(s)$.
		Let the LMC induced by $\alpha$ be $\D(\alpha) = <S, L, \tau, \ell>$.
		
		We define a coupling $\omega_{s,t}$ by setting $\omega_{s,t}(u, v) = w_{s,t,u,v}$ for all $\ell(s)=\ell(t)$ and $u,v \in S$.
		It is easy to verify that $\omega_{s,t}$ is a coupling for the marginals $\tau(s)$ and $\tau(t)$.
		
		We define $d' \in [0,1]^{S \times S}$ by setting $d'(s, t) =d_{s, t}$ for all $s, t \in S$.
		We show that $d'$ is a pre-fixed point of $\Delta$, that is, $\Delta(d') \sqsubseteq d'$.
		Let $s$ and $t$ be two arbitrary states.
		If $\ell(s) = \ell(t)$, we have $\Delta(d')(s, t) = \min_{\omega \in \Omega(\tau(s),\tau(t))}\sum_{u,v \in S}\omega(u, v) d'(u, v) \le \sum_{u,v \in S}\omega_{s, t}(u, v) d_{u,v} = \sum_{u,v \in S}w_{s, t, u, v} d_{u,v}  \QT{=} d_{s, t} = d'(s, t)$.
		If $\ell(s) \neq \ell(t)$, we have $\Delta(d')(s,t) = 1 =  d_{s, t} = d'(s,t)$.
		
		By the Knaster-Tarski's fixed point theorem \cite[Theorem~1]{Tarski1955}, the probabilistic bisimilarity distance $d$ is the least pre-fixed point of $\Delta$.
		Thus, we have $d \sqsubseteq d'$ and $d(s_1, s_2) \le d'(s_1,s_2) = d_{s_1,s_2} \ls \theta$.
		
		\item ``$\impliedby$'':
		Assume there is a memoryless strategy $\alpha : S \to \Dist(\Act)$ such that $d(s_1,s_2) \ls \theta$ in the induced LMC $\D(\alpha) = <S, L, \tau, \ell>$.
		
		By definition of probabilistic bisimilarity distance, we have $\Delta(d) = d$, since $d$ is the least fixed point of the function $\Delta$.
		
		Let $x_{s, \m} = \alpha(s)(\m)$ for all $s \in S$ and $\m \in \Act$ and $d_{s, t} = d(s, t)$ for all $s, t \in S$.
		For all $s, t \in S$ such that $\ell(s) = \ell(t)$, let $\omega_{s,t} \in \arg\min_{\omega \in \Omega(\tau(s), \tau(t))} \sum_{u, v \in S} \omega(u, v) \; d(u, v)$ and $w_{s,t,u,v} = \omega_{s,t}(u, v)$ for all $u, v \in S$.
		
		Thus, the assignments of $(x_{s, \m})_{s \in S, \m \in \Act(s)}$, $(w_{s,t,u,v})_{\ell(s)=\ell(t) \text{ and } u,v \in S}$ and $(d_{s,t})_{s,t \in S}$ witness the validity of the formula as the following items hold:
		\begin{enumerate}
			\item[-]
			$\sum_{\m \in \Act(s)}x_{s, \m} = \sum_{\m \in \Act(s)}\alpha(s)(\m) = 1$ for all $s \in S$;
			\item[-]
			$\sum_{v \in S}w_{s, t, u, v} = \sum_{v \in S}\omega_{s, t}(u, v) = \tau(s)(u)$ for all $\ell(s) = \ell(t)$ and all $u \in S$, since $\omega_{s, t}$ has left marginal $\tau(s)(u)$;
			\item[-]
			$\sum_{u \in S}w_{s, t, u, v} = \sum_{u \in S}\omega_{s, t}(u, v) = \tau(t)(v)$
			for all $s, t \in S$ such that $\ell(s) = \ell(t)$ and all $v \in S$, since $\omega_{s, t}$ has right marginal $\tau(t)(v)$;
			\item[-]
			for all $s,t \in S$ such that $\ell(s) = \ell(t)$, we have
			$\sum_{u,v \in S}w_{s, t, u, v} d_{u,v} = \sum_{u,v \in S}\omega_{s, t}(u, v) d(u,v) = \min_{\omega \in \Omega(\tau(s),\tau(t))}\sum_{u,v \in S}\omega(u, v) d(u,v) = \Delta(d)(s,t) = d(s, t) = d_{s, t}$;
			\item[-]
			$d_{s,t} = d(s, t) = 1$ for all $s,t \in S$ such that $\ell(s) \neq \ell(t)$;
			\item[-]
			$d_{s_1,s_2} =  d(s_1, s_2) \ls \theta$. \qedhere
		\end{enumerate}
	\end{itemize}
	
	\subsection{Proof of \texorpdfstring{\cref{lem:general-distance-expression}}{Lemma \ref{lem:general-distance-expression}}} \label{proof:lem:general-distance-expression}
	Let $\alpha$ be a general strategy.
	\begin{itemize}
		\item ``$\le$'':
		Let the LMC induced by $\alpha$ be $\D(\alpha) = <\mathcal{P}, L, \tau, \ell'>$.
		Let $\rho_1, \rho_2 \in \mathcal{P}$ be two states in $\D(\alpha)$.
		Since the second MDP is an LMC, we simply use the end state of a path that starts with $s_2$.
		That is, we write $\rho(|\rho|)$ for a state $\rho \in \mathcal{P}$ in $\D(\alpha)$ that starts with $s_2$.
		We define $d \in [0,1]^{\mathcal{P} \times \mathcal{P}}$ and distinguish the following cases:
		\begin{itemize}
			\item
			If $\rho_1 = \rho_2$,
			we define $d(\rho_1, \rho_2) = 0$.
			\item
			If $\rho_1$ ends with $x$ (resp. $y$) and $\rho_2$ ends with $x'$ (resp. $y'$),
			we define $d(\rho_1, \rho_2) = d(\rho_2, \rho_1) = 0$.
			\item
			Assume that $\rho_1= s_1v$ is a state that starts with $s_1$ and does not visit the $\$$ state. We have $v \in L^{*}$.
			Assume that $\rho_2$ is a state that starts with $s_2$ and has $av$ as its sequence of labels.
			Let the end state of $\rho_2$ be $(\sigma, q)$.
			We define $d(\rho_1, \rho_2) =  d\Big(s_1v, \big(\sigma, q\big)\Big) = \sum_{w \in L^*} \frac{1}{3^{|w| + 1}} \Big( \big(1-\Pr_{\A_{q}}(w)\big) \alpha(s_1vw\$)(\m_y) + \Pr_{\A_q}(w) \alpha(s_1 v w \$)(\m_x) \Big)$.
			We also define $d(\rho_2, \rho_1) = d(\rho_1, \rho_2)$.
			\item
			Assume that $\rho_1= s_1v\$$ is a state that starts with $s_1$ and ends with the $\$$ state. We have  $v \in L^{*}$.
			Assume that $\rho_2$ is a state that starts with $s_2$ and has $av\$$ as its sequence of labels.
			If the last state of $\rho_2$ is $\$_x$, define $d(\rho_1, \rho_2) =  d(s_1v\$, \$_x) = \alpha(s_1 v \$)(\m_y)$.
			Otherwise, if the last state of $\rho_2$ is $\$_y$, define $d(\rho_1, \rho_2) =  d(s_1v\$, \$_y) =  \alpha(s_1 v \$)(\m_x)$.
			\item
			For all the other cases, we define $d(\rho_1, \rho_2) = 1$.
		\end{itemize}
		
		Next, we show $d$ is a pre-fixed point of $\Delta$.
		By the Knaster-Tarski's fixed point theorem \cite[Theorem~1]{Tarski1955}, the probabilistic bisimilarity distance $d_\alpha$ is the least pre-fixed point of $\Delta$.
		Thus, we have $d_\alpha \sqsubseteq d$ and $d_\alpha(s_1, s_2) \le d(s_1,s_2)  = d\Big(a,\big(a, q_0\big)\Big) = \sum_{w \in L^*} \frac{1}{3^{|w| + 1}} \Big( \big(1-\Pr_{\A}(w)\big) \alpha(s_1 w\$)(\m_y) + \Pr_{\A}(w) \alpha(s_1 w \$)(\m_x) \Big)$.
		
		Let $\rho_1, \rho_2 \in \mathcal{P}$ be two states in the LMC $\D(\alpha)$. We have:
		\begin{itemize}
			\item
			If $\rho_1 = \rho_2$, we have
			\begin{align*}
				&\Delta(d)(\rho_1, \rho_2) \\
				&= \ \min_{\omega \in \Omega(\tau(\rho_1), \tau(\rho_2))}\sum_{\rho_1', \rho_2' \in \mathcal{P}} \omega(\rho_1', \rho_2') d(\rho_1', \rho_2') \\
				&\le \ \sum_{\rho_1', \rho_2' \in \mathcal{P}} \omega'(\rho_1', \rho_2') d(\rho_1', \rho_2') \\
				&\commenteq{Define $\omega' \in \Omega(\tau(\rho_1), \tau(\rho_2))$ by setting $\omega'(\rho_1', \rho_1') = \tau(\rho_1)(\rho_1')$ for all $\rho_1' \in \mathcal{P}$}\\
				&= \ 0 \commenteq{$d(\rho_1', \rho_2') = 0$ for all $\rho_1' = \rho_2'$}\\
				&= \ d(\rho_1, \rho_2)
			\end{align*}
			\item
			If $\rho_1$ ends with $x$ and $\rho_2$ ends with $x'$,
			we have
			\begin{align*}
				&\Delta(d)(\rho_1, \rho_2) \\
				&= \ \min_{\omega \in \Omega(\tau(\rho_1), \tau(\rho_2))}\sum_{\rho_1', \rho_2' \in \mathcal{P}} \omega(\rho_1', \rho_2') d(\rho_1', \rho_2') \\
				&= \ d(\rho_1x, \rho_2x') \\
				&= \ 0 \commenteq{Definition of $d$}\\
				&= \  d(\rho_1, \rho_2)
			\end{align*}
			Similarly, we have $\Delta(d)(\rho_2, \rho_1) = 0 = d(\rho_2, \rho_1)$.
			The case for $\rho_1$ ending with $y$ and $\rho_2$ ending with $y'$ is similar, that is,
			$\Delta(d)(\rho_1, \rho_2) = 0 = d(\rho_1, \rho_2)$ and $\Delta(d)(\rho_2, \rho_1) = 0 = d(\rho_2, \rho_1)$.
			\item
			Assume that $\rho_1= s_1v$ is a state that starts with $s_1$ and does not visit the $\$$ state. We have $v \in L^{*}$.
			Assume that $\rho_2$ is a path that starts with $s_2$ and has $av$ as its sequence of labels.
			Let the end state of $\rho_2$ be $\big(\sigma, q\big)$.
			
			If $q \not\in F$, we have
			\begin{align*}
				&\Delta(d)(\rho_1, \rho_2) \\
				&= \ \Delta(d)\Big(s_1v, \big(\sigma, q\big)\Big) \\
				&= \ \min_{\omega \in \Omega(\tau(s_1v), \tau((\sigma, q)))}\sum_{\rho_1', \rho_2' \in \mathcal{P}} \omega(\rho_1', \rho_2') d(\rho_1', \rho_2') \\
				&\le \ \sum_{\sigma' \in L}\sum_{q' \in Q}\omega'(s_1v\sigma', (\sigma', q')) d(s_1v\sigma', (\sigma', q')) + \omega'(s_1v\$, \$_x) d(s_1v\$, \$_x)\\
				&= \ \sum_{\sigma' \in L}\sum_{q' \in Q}\frac{1}{3}\delta(q, \sigma')(q')d(s_1v\sigma', (\sigma', q')) + \frac{1}{3}d(s_1v\$, \$_x)\\
				&\commenteq{Define $\omega' \in \Omega(\tau(\rho_1), \tau(\rho_2))$ by setting }\\
				&\commenteq{$\omega'(sv\sigma', (\sigma',q'))
					= \frac{1}{3}\delta(q, \sigma')(q')$ for all $\sigma' \in L$ and $q' \in Q$, and}\\
				&\commenteq{$\omega'(sv\$, \$_x) = \frac{1}{3}$}\\
				&= \ \sum_{\sigma' \in L}\sum_{q' \in Q}\frac{1}{3}\delta(q, \sigma')(q')
				\sum_{w \in L^*} \frac{1}{3^{|w| + 1}} \big( (1-\textstyle\Pr_{\A_{q'}}(w)) \alpha(s_1v\sigma' w\$)(\m_y) \ + \\
				&\phantom{=} \ \textstyle\Pr_{\A_{q'}}(w) \alpha(s_1 v\sigma' w \$)(\m_x) \big) \ + \\
				&\phantom{=} \ \frac{1}{3}\alpha(s_1v\$)(\m_y)\\
				&= \ \sum_{w \in L^{+}}\frac{1}{3^{|w| + 1}}\big( (1-\textstyle\Pr_{\A_{q}}(w)) \alpha(s_1vw\$)(\m_y) \ + \\
				&\phantom{=} \ \textstyle\Pr_{\A_{q}}(w) \alpha(s_1 v w \$)(\m_x) \big) \ + \\
				&\phantom{=} \ \textstyle\frac{1}{3^{|\varepsilon| + 1}} \big( (1-\Pr_{\A_{q}}(\varepsilon)) \alpha(s_1v\$)(\m_y) + \Pr_{\A_{q}}(\varepsilon) \alpha(s_1 v \$)(\m_x) \big) \\
				&\commenteq{$|\varepsilon| = 0$ and $\Pr_{\A_{q}}(\varepsilon) = 0$ as $q \not\in F$}\\
				&= \ \sum_{w \in L^{*}}\frac{1}{3^{|w| + 1}}\big( (1-\textstyle\Pr_{\A_{q}}(w)) \alpha(s_1vw\$)(\m_y) + \Pr_{\A_{q}}(w) \alpha(s_1 v w \$)(\m_x) \big) \\
				&= \ d\big(s_1v, (\sigma, q)\big) \commenteq{Definition of $d$}\\
				&= \ d(\rho_1, \rho_2)
			\end{align*}
			We also have $\Delta(d)(\rho_2, \rho_1) = \Delta(d)(\rho_1, \rho_2) \le d(\rho_1, \rho_2) = d(\rho_2, \rho_1)$.
			Similarly, if $q \in F$, we have
			$\Delta(d)(\rho_2, \rho_1) = \Delta(d)(\rho_1, \rho_2) \le d(\rho_1, \rho_2) = d(\rho_2, \rho_1)$.
			\item
			Assume $\rho_1= s_1v\$$ is a path that starts with $s_1$ and ends with the $\$$ state. We have $v \in L^{*}$.
			Assume $\rho_2$ is a path that starts with $s_2$ and has $av\$$ as its sequence of labels.
			If the last state of $\rho_2$ is $\$_x$,
			we have
			\begin{align*}
				&\Delta(d)(\rho_1, \rho_2) \\
				&= \ \Delta(d)(s_1v\$, \$_x) \\
				&= \ \min_{\omega \in \Omega(\tau(s_1v\$), \tau(\$_x))}\sum_{\rho_1', \rho_2' \in \mathcal{P}} \omega(\rho_1', \rho_2') d(\rho_1', \rho_2') \\
				&= \ \alpha(s_1v\$)(\m_x)d(s_1v\$x, x') + \alpha(s_1v\$)(\m_y)d(s_1v\$y, x') \\
				&= \ \alpha(s_1v\$)(\m_y) \commenteq{$d(s_1v\$x, x') = 0$ and $d(s_1v\$y, x') = 1$} \\
				&= \ d(s_1v\$, \$_x) \commenteq{Definition of $d$} \\
				&= \ d(\rho_1, \rho_2).
			\end{align*}
			We also have $\Delta(d)(\rho_2, \rho_1) = \Delta(d)(\rho_1, \rho_2) = d(\rho_1, \rho_2) = d(\rho_2, \rho_1)$.
			Similarly if the last state of $\rho_2$ is $\$_y$, we also have $\Delta(d)(\rho_2, \rho_1) = \Delta(d)(\rho_1, \rho_2) = d(\rho_1, \rho_2) = d(\rho_2, \rho_1)$.
			\item
			For all the other cases, we have $\Delta(d)(\rho_1, \rho_2) \le d(\rho_1, \rho_2) = 1$.
		\end{itemize}
		Thus, $d$ is a pre-fixed point of $\Delta$.
		
		\item ``$\ge$'':
		We denote by $L^{\le n}$ the set of words of length at most $n$.
		To prove $d_\alpha(s_1, s_2)
		\ge \sum_{w \in L^*} \frac{1}{3^{|w| + 1}} \Big( \big(1-\textstyle\Pr_{\A}(w)\big) \alpha(s_1 w
		\$)(\m_y) + \Pr_{\A}(w) \alpha(s_1 w \$)(\m_x) \Big)$, it suffices to prove  $d_\alpha(s_1, s_2)
		\ge \sum_{w \in L^{\le n}} \frac{1}{3^{|w| + 1}} \big( (1-\textstyle\Pr_{\A}(w)) \alpha(s_1 w
		\$)(\m_y) + \Pr_{\A}(w) \alpha(s_1 w \$)(\m_x) \big)$ for all $n \in \nat$.
		
		Let $n \in \nat$. We have
		\begin{align*}
			&d_{\alpha}(s_1, s_2) \\
			&= \ \frac{1}{3}\sum_{\sigma \in L} \sum_{q \in Q} \delta(q_0, \sigma)(q) d_{\alpha}\big(s_1\sigma, (\sigma, q) \big) + \frac{1}{3}d_{\alpha}(s_1\$, \$_x) \commenteq{$q_0 \not\in F$}\\
			&= \ \frac{1}{3}\sum_{\sigma \in L} \sum_{q \in Q} \delta(q_0, \sigma)(q) d_{\alpha}\big(s_1\sigma, (\sigma, q) \big) \ + \\
			&\phantom{=} \ \frac{1}{3}\alpha(s_1\$)(\m_x)d_{\alpha}(s_1\$x, x') + \frac{1}{3}\alpha(s_1\$)(\m_y)d_{\alpha}(s_1\$y, x')\\
			&= \ \frac{1}{3}\sum_{\sigma \in L} \sum_{q \in Q} \delta(q_0, \sigma)(q) d_{\alpha}\big(s_1\sigma, (\sigma, q) \big) \ + \\
			&\phantom{=} \ \frac{1}{3^{|\varepsilon| + 1}} \big( (1-\textstyle\Pr_{\A}(\varepsilon)) \alpha(s_1\$)(\m_y) + \Pr_{\A}(\varepsilon) \alpha(s_1 \$)(\m_x) \big)\\
			&\phantom{=} \commenteq{$|\varepsilon| = 0$; $\Pr_{\A}(\varepsilon) = 0$ since $q_0 \not\in F$; $d_{\alpha}(s_1\$x, x') = 0$ and $d_{\alpha}(s_1\$y, x') = 1$}\\
			&= \ \frac{1}{3}\sum_{\sigma \in L} \sum_{q \in Q} \delta(q_0, \sigma)(q)
			\Big( \frac{1}{3}\sum_{\sigma' \in L} \sum_{q' \in Q} \delta(q, \sigma')(q') d_{\alpha}\big(s_1\sigma\sigma', (\sigma', q') \big) \Big) \ + \\
			&\phantom{=} \  \frac{1}{3} \sum_{\sigma \in L} \sum_{q \not\in F} \delta(q_0, \sigma)(q) \frac{1}{3}d_{\alpha}(s_1\sigma\$, \$_x) + \frac{1}{3} \sum_{\sigma \in L} \sum_{q \in F} \delta(q_0, \sigma)(q) \frac{1}{3} d_{\alpha}(s_1\sigma\$, \$_y) \ +\\
			&\phantom{=} \ \frac{1}{3^{|\varepsilon| + 1}} \Big( \big(1-\textstyle\Pr_{\A}(\varepsilon)\big) \alpha(s_1\$)(\m_y) + \Pr_{\A}(\varepsilon) \alpha(s_1 \$)(\m_x) \Big)\\
			&= \ \frac{1}{3^{2}}\sum_{w \in L^{2}} \sum_{q \in Q} \delta(q_0, w)(q) d_{\alpha}\big(s_1w, (w(2), q) \big) \ + \\
			&\phantom{=} \  \frac{1}{3^{2 }}\sum_{\sigma \in L}  \sum_{q \not\in F} \delta(q_0, \sigma)(q) \big( \alpha(s_1\sigma\$)(\m_x) d_{\alpha}(s_1\sigma\$x, x') + \alpha(s_1\sigma\$)(\m_y) d_{\alpha}(s_1\sigma\$y, x') \big) \ + \\
			&\phantom{=} \ \frac{1}{3^{2 }}\sum_{\sigma \in L}   \sum_{q \in F} \delta(q_0, \sigma)(q) \big( \alpha(s_1\sigma\$)(\m_x) d_{\alpha}(s_1\sigma\$x, y') + \alpha(s_1\sigma\$)(\m_y) d_{\alpha}(s_1\sigma\$y, y')\big) \ +\\
			&\phantom{=} \ \frac{1}{3^{|\varepsilon| + 1}} \Big( \big(1-\textstyle\Pr_{\A}(\varepsilon)\big) \alpha(s_1\$)(\m_y) + \Pr_{\A}(\varepsilon) \alpha(s_1 \$)(\m_x) \Big)\\
			&= \ \frac{1}{3^{2}}\sum_{w \in L^{2}} \sum_{q \in Q} \delta(q_0, w)(q) d_{\alpha}\big(s_1w, (w(2), q) \big) \ + \\
			&\phantom{=} \  \frac{1}{3^{2 }}\sum_{\sigma \in L}  \sum_{q \not\in F} \delta(q_0, \sigma)(q) \alpha(s_1\sigma\$)(\m_y) +  \frac{1}{3^{2 }}\sum_{\sigma \in L}   \sum_{q \in F} \delta(q_0, \sigma)(q) \alpha(s_1\sigma\$)(\m_x)  \ +\\
			&\phantom{=} \ \frac{1}{3^{|\varepsilon| + 1}} \Big( \big(1-\textstyle\Pr_{\A}(\varepsilon)\big) \alpha(s_1\$)(\m_y) + \Pr_{\A}(\varepsilon) \alpha(s_1 \$)(\m_x) \Big)\\
			&\phantom{=} \commenteq{$d_{\alpha}(s_1\sigma\$x, x') = d_{\alpha}(s_1\sigma\$y, y') = 0$ and $d_{\alpha}(s_1\sigma\$y, x') = d_{\alpha}(s_1\sigma\$x, y') = 1$}\\
			&= \ \frac{1}{3^{2}}\sum_{w \in L^{2}} \sum_{q \in Q} \delta(q_0, w)(q) d_{\alpha}\big(s_1w, (w(2), q) \big) \ + \\
			&\phantom{=} \  \sum_{\sigma \in L}\frac{1}{3^{|\sigma| + 1}} \Big( \big(1-\textstyle\Pr_{\A}(\sigma)\big) \alpha(s_1\sigma\$)(\m_y) + \Pr_{\A}(\sigma) \alpha(s_1\sigma\$)(\m_x) \Big) \ +\\
			&\phantom{=} \ \frac{1}{3^{|\varepsilon| + 1}} \big( (1-\textstyle\Pr_{\A}(\varepsilon)) \alpha(s_1\$)(\m_y) + \Pr_{\A}(\varepsilon) \alpha(s_1 \$)(\m_x) \big)\\
			&= \ \frac{1}{3^{2}}\sum_{w \in L^{2}} \sum_{q \in Q} \delta(q_0, w)(q) d_{\alpha}\big(s_1w, (w(2), q) \big) \ + \\
			&\phantom{=} \  \sum_{w \in L^{\le 1}}\frac{1}{3^{|w| + 1}} \Big( \big(1-\textstyle\Pr_{\A}(w)\big) \alpha(s_1w\$)(\m_y) + \Pr_{\A}(w) \alpha(s_1w\$)(\m_x) \Big) \\
			&= \ \cdots \\
			&= \ \frac{1}{3^{n + 1}}\sum_{w \in L^{n + 1}} \sum_{q \in Q} \delta(q_0, w)(q) d_{\alpha}\big(s_1w, (w(n + 1), q) \big) \ + \\
			&\phantom{=} \ \sum_{w \in L^{\le {n}}}\frac{1}{3^{|w|+1}} \Big( \big(1-\textstyle\Pr_{\A}(w)\big) \alpha(s_1w\$)(\m_y) + \Pr_{\A}(w) \alpha(s_1w\$)(\m_x) \Big)\\
			&\ge \sum_{w \in L^{\le {n}}}\frac{1}{3^{|w|+1}} \Big( \big(1-\textstyle\Pr_{\A}(w)\big) \alpha(s_1w\$)(\m_y) + \Pr_{\A}(w) \alpha(s_1w\$)(\m_x) \Big) \\
			&\phantom{=} \commenteq{$d_{\alpha}\big(s_1w, (w(n + 1), q) \big) \ge 0$ for all $w \in L^{n+1}$} \qedhere
		\end{align*}
	\end{itemize}
	
	\subsection{Proof of \texorpdfstring{\cref{thm:general-undecidable}}{Theorem \ref{thm:general-undecidable}}} \label{proof:thm:general-undecidable}
	Consider the construction from the main body.
	
	We show that there is a word $w \in L^{*}$ such that $\Pr_{\A}(w) \gr \frac{1}{2}$ ($\A$ is nonempty) if and only if there is a general strategy $\alpha$ such that $d_{\alpha}(s_1,s_2) \ls \theta$ in the induced LMC.
	
	\begin{itemize}
		\item
		``$\implies$'':
		Assume there is a word $w \in L^{*}$ such that $\Pr_{\A}(w) \gr \frac{1}{2}$. We show that there exists a general strategy $\alpha$ such that $d_\alpha(s_1,s_2) \ls \theta$ in the LMC induced by $\alpha$.
		
		Define the general strategy $\alpha$ such that $\alpha(s_1w\$) = \mathbf{1}_{\m_y}$ and $\alpha(s_1w'\$) = \mathbf{1}_{\m_x}$ for all $w' \in L^{*} \setminus \{w\}$.
		The strategy $\alpha$ differs from $\alpha_x$ only on the path $s_1w\$$. We have
		\begin{align*}
			&d_\alpha(s_1, s_2) \\
			&=\ \sum_{w' \in L^*} \frac{1}{3^{|w'| + 1}} \Big( \big(1-\textstyle\Pr_{\A}(w')\big) \alpha(s_1 w'
			\$)(\m_y) + \Pr_{\A}(w') \alpha(s_1 w' \$)(\m_x) \Big)\\
			& \commenteq{\cref{lem:general-distance-expression}}\\
			&= \sum_{w' \in L^* \setminus \{w\}} \frac{1}{3^{|w'| + 1}} \Big( \big(1-\textstyle\Pr_{\A}(w')\big) \alpha(s_1 w'\$)(\m_y) + \Pr_{\A}(w') \alpha(s_1 w' \$)(\m_x) \Big) \ +\\
			&\phantom{=} \frac{1}{3^{|w| + 1}} \Big( \big(1-\textstyle\Pr_{\A}(w)\big) \alpha(s_1 w
			\$)(\m_y) + \Pr_{\A}(w) \alpha(s_1 w \$)(\m_x) \Big)\\
			&= \sum_{w' \in L^* \setminus \{w\}} \frac{1}{3^{|w'| + 1}} \Big( \big(1-\textstyle\Pr_{\A}(w')\big) \alpha(s_1 w'\$)(\m_y) + \Pr_{\A}(w') \alpha(s_1 w' \$)(\m_x) \Big) \ + \\
			&\phantom{=} \frac{1}{3^{|w| + 1}} (1-\textstyle\Pr_{\A}(w)) \\
			&\commenteq{$\alpha(s_1w\$) = \mathbf{1}_{\m_y}$}\\
			&\ls  \sum_{w' \in L^* \setminus \{w\}} \frac{1}{3^{|w'| + 1}} \Big( \big(1-\textstyle\Pr_{\A}(w')\big) \alpha(s_1 w'\$)(\m_y) + \Pr_{\A}(w') \alpha(s_1 w' \$)(\m_x) \Big) \ + \\
			&\phantom{=} \frac{1}{3^{|w| + 1}} \textstyle\Pr_{\A}(w)\\
			& \commenteq{$\Pr_{\A}(w) \gr \frac{1}{2}$}\\
			&=  \sum_{w' \in L^* \setminus \{w\}} \frac{1}{3^{|w'| + 1}} \Big( \big(1-\textstyle\Pr_{\A}(w')\big) \alpha(s_1 w'\$)(\m_y) + \Pr_{\A}(w') \alpha(s_1 w' \$)(\m_x) \Big) \ + \\
			&\phantom{=} \frac{1}{3^{|w| + 1}} \Big( \big(1-\textstyle\Pr_{\A}(w)\big) \alpha_x(s_1 w
			\$)(\m_y) + \Pr_{\A}(w) \alpha_x(s_1 w \$)(\m_x) \Big)\\
			&\commenteq{$\alpha_x(s_1w\$) = \mathbf{1}_{\m_x}$}\\
			&=\ \sum_{w' \in L^*} \frac{1}{3^{|w'| + 1}} \Big( \big(1-\textstyle\Pr_{\A}(w')\big) \alpha_x(s_1 w'
			\$)(\m_y) + \Pr_{\A}(w') \alpha_x(s_1 w' \$)(\m_x) \Big)\\
			&= \ d_{\alpha_x}(s_1, s_2) \commenteq{\cref{lem:general-distance-expression}}\\
			&= \ \theta
		\end{align*}
		\item ``$\impliedby$'':
		Assume for all words $w \in L^{*}$ we have $\Pr_{\A}(w) \le \frac{1}{2}$.
		We prove for all general strategy $\alpha$ the distance between $s_1$ and $s_2$ in the induced LMC $\D(\alpha)$ is at least $\theta$. That is, $d_{\alpha}(s_1, s_2) \ge \theta$ for all $\alpha$ .
		
		Let $\alpha$ be an arbitrary general strategy.
		For any word $w \in L^{*}$, we have
		\begin{align}\label{ineq:helper}
			& \big(1-\textstyle\Pr_{\A}(w)\big) \alpha(s_1 w
			\$)(\m_y) + \Pr_{\A}(w) \alpha(s_1 w \$)(\m_x) \nonumber\\
			&\ge \textstyle\Pr_{\A}(w) \commenteq{$\alpha(s_1 w
				\$)(\m_y) + \alpha(s_1 w
				\$)(\m_x) = 1$ and $\Pr_{\A}(w) \le \frac{1}{2}$ for all $w \in L^{*}$} \nonumber\\
			&= \big(1-\textstyle\Pr_{\A}(w)\big) \alpha_x(s_1 w\$)(\m_y) + \Pr_{\A}(w) \alpha_x(s_1 w \$)(\m_x) \\
			& \commenteq{$\alpha_x(s_1 w \$) = \mathbf{1}_{\m_x}$} \nonumber
		\end{align}
		
		We have
		\begin{align*}
			&d_\alpha(s_1, s_2) \\
			&=\ \sum_{w \in L^*} \frac{1}{3^{|w| + 1}} \Big( \big(1-\textstyle\Pr_{\A}(w)\big) \alpha(s_1 w
			\$)(\m_y) + \Pr_{\A}(w) \alpha(s_1 w \$)(\m_x) \Big)\\
			& \commenteq{\cref{lem:general-distance-expression}}\\
			&\ge\ \sum_{w \in L^*} \frac{1}{3^{|w| + 1}} \Big( \big(1-\textstyle\Pr_{\A}(w)\big) \alpha_x(s_1 w
			\$)(\m_y) + \Pr_{\A}(w) \alpha_x(s_1 w \$)(\m_x) \Big)\\
			& \commenteq{Inequality (\ref{ineq:helper})} \\
			&= \ d_{\alpha_x}(s_1, s_2) \commenteq{\cref{lem:general-distance-expression}}\\
			&= \  \theta \qedhere
		\end{align*}
	\end{itemize}
	
	\subsection{Proof of \texorpdfstring{\cref{theorem:reach}}{Theorem \ref{theorem:reach}}} \label{proof:theorem:reach}
	Fix an LMC $<S, L, \tau, \ell>$ where the set $S$ is a countable set of states.

		For $T \in  \mathcal{T}$, the function $\Theta^T : [0, 1]^{S^2} \to [0, 1]^{S^2}$ is defined by
		\[
		\Theta^T(e)(s, t)
		=
		\left \{
		\begin{array}{ll}
			0
			& \mbox{if $(s, t) \in S^2_0$,}\\
			1
			& \mbox{if $(s, t) \in S^2_1$,}\\
			\displaystyle \sum_{u, v \in S} T(s, t)(u, v) \, e(u, v)
			& \mbox{otherwise.}
		\end{array}
		\right.
		\]
	The reachability probabilities $\R_\M^T(\cdot,S^2_1)$ form the least fixed point of~$\Theta^T$:
	\begin{lemma} \label{lemma:reach=lfp}
		For all $T \in \mathcal{T}$ we have $\R_\M^T(\cdot,S^2_1) = \lfp.\Theta^T$, where $\lfp.\Theta^T$ denotes the least fixed point of~$\Theta^T$.
	\end{lemma}
	\begin{proof}
		Let $T \in \mathcal{T}$.
		For $s,t \in S$ and $n \in \{0, 1, \ldots\}$ write $\R((s,t),n) \in [0,1]$ for the probability that in the Markov chain~$\C_{\M}^T$ the pair $(s,t)$ reaches a state in~$S^2_1$ in exactly $n$~steps.
		Also write $\theta^{(n)}$ for the $n$-fold application of $\Theta^T$ to the constant zero function; i.e., $\theta^{(0)}((s,t)) = 0$ for all $(s,t) \in S^2$, and $\theta^{(n+1)} = \Theta^T(\theta^{(n)})$.
		It is easy to show by induction that for all $n \ge 0$ we have $\sum_{i=0}^n \R(\cdot,i) = \theta^{(n+1)}$.
		Thus, $\R_\M^T(\cdot,S^2_1) = \sum_{i=0}^\infty \R(\cdot,i) = \lim_{n \to \infty} \theta^{(n)} = \lfp.\Theta^T$, where the last equality follows from the Kleene fixed point theorem.
	\end{proof}
	Next we show that $\lfp.\Theta^T$ is a pre-fixed point of $\Delta$.
	\begin{proposition}\label{proposition:theta-prefixpoint-delta}
		For all $T \in \mathcal{T}$, $\Delta(\lfp.\Theta^T) \sqsubseteq \lfp.\Theta^T$.
	\end{proposition}
	\begin{proof}
		Let $s, t \in S$.  We distinguish three cases.
		\begin{itemize}
			\item
			If $(s,t) \in  S_1^2$ then
			$\Delta(\lfp.\Theta^T)(s, t) = 1 = \lfp.\Theta^T(s, t)$.
			
			\item
			If $(s, t) \in  S^2_0$, then $s \sim t$. We have $\ell(s) = \ell(t)$, and there exists an $\pi \in \Omega(\tau(s), \tau(t))$ such that $\support(\pi) \subseteq \mathord{\sim}$.
			\[
			\begin{array}{lll}
				\Delta(\lfp.\Theta^{T})(s,t) &=& \displaystyle\min_{\omega \in \Omega(\tau(s), \tau(t))}\sum_{(u, v) \in S^2} \omega(u, v) \, \lfp.\Theta^{T}(u, v) \\
				& \le & \displaystyle\sum_{(u, v) \in S^2} \pi(u, v) \, \lfp.\Theta^{T}(u, v) \\
				& = & \displaystyle\sum_{(u, v) \in \support(\pi)} \pi(u, v) \, \lfp.\Theta^{T}(u, v) \\
				& = & 0 \\
				&&\commenteq{$\forall u \sim v: \lfp.\Theta^{T}(u, v) = \Theta^{T}(\lfp.\Theta^{T})(u, v) = 0$}\\
				& = & \Theta^{T}(\lfp.\Theta^{T})(s, t) \commenteq{since $s \sim t$, $\Theta^{T}(\lfp.\Theta^{T})(s, t)  = 0$}\\
				& = & \lfp.\Theta^{T}(s, t).
			\end{array}
			\]
			
			\item
			Otherwise,
			\begin{eqnarray*}
				\Delta(\lfp.\Theta^T)(s, t)
				& = & \min_{\omega \in \Omega(\tau(s), \tau(t))} \sum_{u, v \in S} \omega(u, v) \, \lfp.\Theta^T(u, v)\\
				& \leq & \sum_{u, v \in S} T(s, t)(u, v) \, \lfp.\Theta^T(u, v)
				\commenteq{$T(s, t) \in \Omega(\tau(s), \tau(t))$}\\
				& = & \Theta^T(\lfp.\Theta^T)(s, t)\\
				& = & \lfp.\Theta^T(s, t).
			\end{eqnarray*}
		\end{itemize}
	\end{proof}
	
	By the Knaster-Tarski's fixed point theorem \cite[Theorem~1]{Tarski1955}, the probabilistic bisimilarity distance $d$ is the least pre-fixed point of $\Delta$.
	Together with \cref{proposition:theta-prefixpoint-delta}, we have:
	
	\begin{corollary}\label{cor:delta-smaller-than-theta}
		For all $T \in \mathcal{T}$, $d \sqsubseteq \lfp.\Theta^T$.
	\end{corollary}

	\begin{proposition}\label{proposition:delta-equal-to-theta}
		There exists a $T \in \mathcal{T}$ such that $\lfp.\Theta^T \sqsubseteq d$.
	\end{proposition}
	\begin{proof}
		For each $(s, t) \in S^2_?$, define
		$
		T(s, t) \in \displaystyle\argmin_{\omega \in \Omega(\tau(s), \tau(t))} \sum_{u, v \in S} \omega(u, v) \, d(u, v).
		$	
		Obviously, $T \in \mathcal{T}$.
		
		Since $\lfp.\Theta^T$ is the least fixed point of $\Theta^T$, it suffices to show that $d$ is a fixed point of $\Theta^T$.
		Let $s, t \in S$.
		We distinguish three cases.
		\begin{itemize}
			\item
			If $(s, t) \in S^2_1$, then
			$\Theta^{T}(d)(s, t) = 1 = d(s, t)$.
			\item
			If $(s, t) \in S^2_0$ then $s \sim t$. We have
			$\Theta^{T}(d)(s, t) = 0 = d(s, t)$.
			\item
			Otherwise,
			\begin{eqnarray*}
				\Theta^T(d)(s, t)
				& = & \sum_{u, v \in S} T(s, t)(u, v) \, d(u, v)\\
				& = & \min_{\omega \in \Omega(\tau(s), \tau(t))} \sum_{u, v \in S} \omega(u, v) \, d(u, v)\\
				&&\commenteq{by construction of $T$}\\
				& = & \Delta(d)(s, t)\\
				& = & d(s, t).
			\end{eqnarray*}
		\end{itemize}
	\end{proof}

	It follows from \cref{cor:delta-smaller-than-theta} and \cref{proposition:delta-equal-to-theta} that
	$\displaystyle d = \min_{T \in \mathcal{T}} \lfp.\Theta^T$.
	Combined with \cref{lemma:reach=lfp}, this implies \cref{theorem:reach}.
	
	\subsection{Proof of \texorpdfstring{\cref{prop:mdp-distance-ls1}}{Proposition \ref{prop:mdp-distance-ls1}}} \label{proof:prop:mdp-distance-ls1}

	\begin{lemma} \label{lemma:inf-LMC-game-step}
		Let $\M_1, \M_2, \ldots$ be an infinite sequence of LMCs with $\M_i = <S, L, \tau_i, \ell>$ for all~$i$.
		Suppose that for each $s \in S$ there is a finite set $N(s) \subseteq S$ with $\support(\tau_i(s)) \subseteq N(s)$ for all~$i$.
		Let $S_0 \subseteq S$  be a finite subset with $\liminf_{i \to \infty} \max\{d_{\M_i}(s,t) \mid s,t \in S_0\} = 0$; i.e., there is a subsequence such that all pairwise distances in~$S_0$ converge to~$0$.
		Write $N := \bigcup_{s \in S_0} N(s)$.
		Then there are a transition function $\tau: S_0 \to \Dist(S)$, a partition $\{S_1, \ldots, S_k\}$ of~$N$ and a probability distribution~$\mu$ on $\{1, \ldots, k\}$
		such that
		$\tau(s)(S_j) = \mu(j)$ for all $s \in S_0$ and $j \in \{1, \ldots, k\}$ and
		\begin{align*}
			\liminf_{i \to \infty} \max  & \left\{|\tau_i(s) - \tau(s)| \mid s \in S_0 \right\} \\
			& \; \mbox{} \cup \left\{ |\tau_i(s)(S_j) - \mu(j)| \mid s \in S_0,\ 1 \le j \le k \right\} \\
			&  \; \mbox{} \cup \left\{ d_{\M_i}(s,t) \mid 1 \le j \le k,\ s,t \in S_j \right\} \quad = \quad 0\,;
		\end{align*}
		i.e., there is a subsequence of $\M_1, \M_2, \ldots$ such that for all $s \in S_0$ and all $1 \le j \le k$ the transition function converges to~$\tau$, the transition probability into~$S_j$ converges to~$\mu(j)$, and for each $S_j$ with $1 \le j \le k$ all pairwise distances in~$S_j$ converge to~$0$.
	\end{lemma}
	\begin{proof}
		For a vector $\mu \in \mathbb{R}^{X}$ where $X$ is a finite set, we write $|\mu| := \sum_{x \in S} |\mu(x)|$ for its  $L_1$ norm.
		
		Let $n$ be the number of states in $S_0$ and write $S_0 = \{s_1, \ldots, s_n\}$.
		For each $s \in S_0$ and all~$\M_i$ we have $\support(\tau_i(s)) \subseteq N$.
		So we can view each $\tau_i(s)$ as an element of $[0,1]^N$.
		For $1 \le p \ls n$ we denote by $\omega_{i}^{p} \in [0,1]^{N^2}$ the optimal coupling with marginals $\tau_i(s_p)$ and $\tau_i(s_{p+1})$; i.e., $d_{\M_i}(s_p,s_{p+1}) = \sum_{u, v \in N} \omega_i^p(u, v)d_{\M_i}(u,v)$.
		
		For each~$\M_i$, we can view $(\tau_i(s_1), \ldots, \tau_i(s_n), \omega_i^1, \ldots, \omega_i^{n-1})$ as an element of the compact set $[0,1]^{n |N| + (n-1) |N|^2}$.
		Therefore, by the Bolzano-Weierstrass theorem and the assumption, there exist a (limit) transition function $\tau: S_0 \to \Dist(S)$ (we have $\support(\tau(s)) \subseteq N$ for all $s \in S_0$) and $n-1$ (limit) couplings $\omega^{p}: [0, 1]^{N^2}$ with $1 \le p \ls n$ such that
		\begin{align}\label{eq:subseq}
			\liminf_{i \to \infty} \max &\left\{d_{\M_i}(s,t) \mid s,t \in S_0\right\} \nonumber \\
			& \; \mbox{} \cup  \left\{|\tau_i(s) - \tau(s)| \mid s \in S_0 \right\} \nonumber \\
			& \; \mbox{} \cup  \left\{ |\omega_i^p - \omega^p | \mid 1 \le p \ls n\right\} \ = \ 0;
		\end{align}
		i.e., there is a subsequence of $\M_1, \M_2, \ldots$, say $\M_{i_1}, \M_{i_2}, \ldots$, such that all pairwise distances in~$S_0$ converge to~$0$, the transition functions $\tau_{i_j}$ converge to $\tau$ and the optimal coupling $\omega_{i_j}^p$ converge to~$\omega^p$.
		
		We show next that the limit coupling $\omega^p$ has left marginal $\tau(s_p)$ and right marginal $\tau(s_{p+1})$ for all $1 \le p \ls n$;
		i.e.,
		\begin{equation}\label{eq:left-marg}
			\sum_{v \in S}\omega^p(u, v) = \tau(s_p)(u)
		\end{equation} and
		\begin{equation}\label{eq:right-marg}
			\sum_{u \in S}\omega^p(u, v) = \tau(s_{p+1})(v).
		\end{equation}
		
		We first show that the left marginal of $\omega^p$ is $\tau(s_p)$.
		We assume towards a contradiction that there is a state $u \in S$ such that $\varepsilon = \tau(s_p)(u)-\sum_{v \in S} \omega^p(u, v) \gr 0$. (The case $\varepsilon \ls 0$ is similar.)
		By \eqref{eq:subseq}, there exists $m_1$ such that for all $m \gr m_1$ we have $|\tau_{i_m}(s_p) - \tau(s_p)| \ls \varepsilon/2$. There also exists $m_2$ such that for all $m \gr m_2$ we have $|\omega_{i_{m}} - \omega^p| \ls \varepsilon/2$. 
		Let $m' = \max\{m_1, m_2\}$ and $m \gr m'$.
		We have the following contradiction:
		\begin{align*}
			& \; \varepsilon \\
			= & \; \tau(s_p)(u) - \sum_{v \in S} \omega^p(u, v) \\
			= & \; \tau(s_p)(u) - \sum_{v \in S} \omega^p(u, v) - \tau_{i_m}(s_p)(u) +  \sum_{v \in S} \omega_{i_m}^p(u, v) \\
			&\commenteq{$\tau_{i_m}(s_p)$ is the left marginal of $\omega_{i_m}^p$}\\
			= & \; \big( \tau(s_p)(u) - \tau_{i_m}(s_p)(u) \big) +  \sum_{v \in S} \big( \omega_{i_m}^{p}(u, v) - \omega^{p}(u, v) \big) \\
			\le & \; |\tau(s_p) - \tau_{i_m}(s_p)| +  | \omega_{i_m}^{p} - \omega^{p}| \\
			\ls & \; \frac{\varepsilon}{2}  + \frac{\varepsilon}{2} \\
			= & \; \varepsilon
		\end{align*}
		The case for $\tau(s_{p+1})$ being the right marginal of $\omega^p$ is similar.
		
		Let $R$ be the relation on~$N$ with $R := \bigcup_{1 \le p \ls n} \support(\omega^p)$.
		Let $\bar{R}$ be the reflexive, symmetric and transitive closure of $R$.
		Then $\bar{R}$ is an equivalence relation on~$N$.
		Denote by $S_1, \ldots, S_k$ the equivalence classes of~$\bar{R}$, so that $\{S_1, \ldots, S_k\}$ is a partition of~$N$.

		For all $1 \le p \ls n$ and $1 \le j \le k$ we have
		\begin{align*}
			& \; \tau(s_p)(S_j) \\
			= & \; \sum_{u \in S_j}\sum_{v \in S }\omega^p(u, v) \commenteq{by \eqref{eq:left-marg}}\\
			= & \; \sum_{u \in S_j}\sum_{v \in S_j }\omega^p(u, v) \\
			& \; \commenteq{$u$ and $v$ are in the same equivalence class if $(u, v) \in \support(\omega^p)$}\\
			= & \; \sum_{v \in S_j}\sum_{u \in S}\omega^p(u, v) \\
			= & \; \sum_{v \in S_j}\tau(s_{p+1})(v)
			\commenteq{by \eqref{eq:right-marg}}\\
			= & \; \tau(s_{p+1})(S_j) \,.
		\end{align*}
		Thus, we have $\tau(s)(S_j) = \tau(t)(S_j)$ for all $s, t \in S_0$ and $1 \le j \le k$.
		Hence, there is $\mu \in \Dist(\{1, \ldots, k\})$ with $\mu(j) = \tau(s)(S_j)$ for all $s \in S_0$ and all $1 \le j \le k$.
		It follows that in the subsequence $\M_{i_1}, \M_{i_2}, \ldots$ the transition probability into $S_j$ converges to $\mu(j)$, since for any $s \in S_0$ and $1 \le j \le k$, we have
		$|\tau_{i}(s)(S_j) - \mu(j)| = |\tau_{i}(s)(S_j) - \tau(s)(S_j)| \le |\tau_{i}(s) - \tau(s)|$.
		
		Finally, we show that in the subsequence $\M_{i_1}, \M_{i_2}, \ldots$, for each $S_j$ with $1 \le j \le k$ all pairwise distances in $S_j$ converge to $0$.
		Let $1 \le j \le k$ and $(u,v) \in S_j$.
		We need to show that $d_{\M_{i_1}}(u,v), d_{\M_{i_2}}(u,v), \ldots$ converges to~$0$.
		It follows from the definition of~$S_j$ and the triangle inequality that we can assume that there is $1 \le p \ls n$ with $(u, v) \in \support(\omega^p)$, i.e., $\omega^p(u,v) \gr 0$.
		Since $\omega_{i_1}^p, \omega_{i_2}^p, \ldots$ converges to $\omega^p$,
		there exists $m_3$ such that for all $m \gr m_3$ we have $\omega^p(u, v) - \omega_{i_m}^p(u, v) \le \omega^p(u, v)/2$.
		Therefore, $\omega_{i_m}^p(u, v) \ge \omega^p(u, v)/2$ for all $m \gr m_3$.
		Towards a contradiction assume that $\limsup_{j \to \infty} d_{\M_{i_j}}(u,v) \gr 0$.
		Then we have
		\begin{align*}
			& \limsup_{j \to \infty} d_{\M_{i_j}}(s_p,s_{p+1}) \ \ge \ \limsup_{j \to \infty} \omega_{i_j}^p(u,v) d_{\M_{i_j}}(u,v) \ \ge \
			\frac{\omega^p(u,v)}2 \limsup_{j \to \infty} d_{\M_{i_j}}(u,v) \\
			&\gr \ 0\,,
		\end{align*}
		contradicting the fact that $\lim_{j \to \infty} d_{\M_{i_j}}(s_p,s_{p+1}) = 0$.
		Therefore, we conclude that $\lim_{j \to \infty} d_{\M_{i_j}}(u,v) = 0$, as desired.
	\end{proof}

	\begin{lemma} \label{lemma:MDP-game-step}
		Let $\D = <S, \Act, L, \varphi, \ell>$ be an MDP.
		Let $\alpha_1, \alpha_2, \ldots$ be an infinite sequence of strategies and $S_0 \subseteq S$ be a finite subset with $\liminf_{i \to \infty} \max \{d_{\D(\alpha_i)} (s, t) \mid s, t \in S_0)\} = 0$.
		Let $\tau_i$ be the probability transition function for the induced LMC $\D(\alpha_i)$, i.e., $\tau_i(\rho)(\rho\m t) = \alpha_i(\rho)(\m)\varphi(s, \m)(t)$ for a state $\rho$ in $\D(\alpha_i)$.
		For $\rho \in \Paths(\D)$, define $N(\rho) = \left\{\rho \m t \mid s=\last(\rho), \m \in \Act(s), t \in \support(\varphi(s)(\m)) \right\}$.
		Write $N := \bigcup_{s \in S_0} N(s)$.
		Then there are a strategy $\alpha: S_0 \to \Dist(\Act)$, a transition function $\tau: S_0 \to \Dist(\Paths(\D))$, a partition $\{S_1, \ldots, S_k\}$ of $N$ and
		a probability distribution $\mu$ on $\{1, \ldots, k\}$ such that
		$\tau(s)(s\m t) = \alpha(s)(\m)\varphi(s,\m)(t)$ for all $s \in S_0$ and $s\m t \in N$,
		$\tau(s)(S_j) = \mu(j)$ for all $s \in S_0$ and $j \in \{1, \ldots, k\}$,
		and
		\begin{align*}
			\liminf_{i \to \infty} \max & \left\{ |\alpha_i(s) - \alpha(s)| \mid s \in S_0 \right\} \\
			& \; \mbox{} \cup \left\{ |\tau_i(s) - \tau(s)| \mid s \in S_0 \right\} \\
			& \; \mbox{} \cup \left\{ |\tau_i(s)(S_j) - \mu(j) | \mid s \in S_0,\ 1 \le j \le k \right\} \\
			& \; \mbox{} \cup \left\{ d_{\D(\alpha_i)}(\rho_1,\rho_2) \mid 1 \le j \le k,\ \rho_1,\rho_2 \in S_j \right\} \quad = \quad 0\,;
		\end{align*}
		i.e., there is a subsequence of $\D(\alpha_1), \D(\alpha_2), \ldots$ such that for all $s \in S_0$ and all $1 \le j \le k$ the strategy converges to $\alpha$, the transition function converges to~$\tau$, the transition probability into~$S_j$ converges to~$\mu(j)$, and for each $S_j$ all pairwise distances in~$S_j$ converge to~$0$.
	\end{lemma}
	\begin{proof}
		We have that $\D(\alpha_1), \D(\alpha_2), \ldots$ is an infinite sequence of LMCs.
		Since the MDP $\D$ is finite, we have that $N(\rho)$ is finite and $N(\rho) \subseteq \Paths(\D)$.
		
		By the assumption that $\liminf_{i \to \infty} \max \{d_{\D(\alpha_i)} (s, t) \mid s, t \in S_0\} = 0$ and \cref{lemma:inf-LMC-game-step}, there are a partition $\{S_1, \ldots, S_k\}$ of $N$,
		a probability distribution $\mu$ on $\{1, \ldots, k\}$ and a transition function $\tau: S_0 \to \Dist(\Paths(\D))$ such that
		$\tau(s)(S_j) = \mu(j)$ for all  $s \in S_0$ and $j \in \{1, \ldots, k\}$,
		and
		\begin{align}
			\liminf_{i \to \infty} \max & \left\{ |\tau_i(s) - \tau(s)| \mid s \in S_0 \right\} \nonumber\\
			& \; \mbox{} \cup \left\{ |\tau_i(s)(S_j) - \mu(j) | \mid s \in S_0,\ 1 \le j \le k \right\} \nonumber\\
			& \; \mbox{} \cup \left\{ d_{\D(\alpha_i)}(\rho_1,\rho_2) \mid 1 \le j \le k,\ \rho_1,\rho_2 \in S_j \right\} \quad = \quad 0\,\label{eq:MDP-game-step};
		\end{align}
		i.e., there is a subsequence of $\D(\alpha_1), \D(\alpha_2), \ldots$ such that for all $s \in S_0$ and all $1 \le j \le k$ the transition function converges to $\tau$, the transition probability into~$S_j$ converges to~$\mu(j)$, and for each $S_j$ all pairwise distances in~$S_j$ converge to~$0$.
		
		Let $n$ be the number of states in $S_0$ and write $S_0 = \{s_1, \ldots, s_n\}$.
		For each $s \in S_0$ and all~$\alpha_i$ we have $\support(\alpha_i(s)) \subseteq \Act$.
		So we can view each $\alpha_i(s)$ as an element of $[0,1]^{|\Act|}$.
		For each~$\alpha_i$, we can view $(\alpha_i(s_1), \ldots, \alpha_i(s_n))$ as an element of the compact set $[0,1]^{n|Act|}$.
		Therefore, by the Bolzano-Weierstrass theorem and \eqref{eq:MDP-game-step}, there exists a (limit) strategy $\alpha: S_0 \to \Dist(\Act)$ such that
		\begin{align}\label{eq:subseqmdp}
			\liminf_{i \to \infty} \max & \left\{ |\tau_i(s) - \tau(s)| \mid s \in S_0 \right\} \nonumber\\
			& \; \mbox{} \cup \left\{ |\tau_i(s)(S_j) - \mu(j) | \mid s \in S_0,\ 1 \le j \le k \right\} \nonumber\\
			& \; \mbox{} \cup \left\{ d_{\D(\alpha_i)}(\rho_1,\rho_2) \mid 1 \le j \le k,\ \rho_1,\rho_2 \in S_j \right\} \nonumber\\
			& \; \mbox{} \cup \left\{ |\alpha_i(s) - \alpha(s)| \mid s \in S_0 \right\} \quad = \quad 0\,;
		\end{align}
		i.e., there is a subsequence of $\D(\alpha_1), \D(\alpha_2), \ldots$ such that for all $s \in S_0$ and all $1 \le j \le k$ the transition function converges to $\tau$, the transition probability into~$S_j$ converges to~$\mu(j)$, for each $S_j$ all pairwise distances in~$S_j$ converge to~$0$ and the strategy converges to $\alpha$.
		
		It remains to show that the limit strategy $\alpha$ on $S_0$ is compatible with the limit transition function $\tau$, that is, $\tau(s)(s \m t) = \alpha(s)(\m)\varphi(s, \m)(t)$ for all $s \in S_0$ and $s \m t \in N$.
		
		We assume towards a contradiction that there is a state $s \in S_0$ and $s \m t \in N$ such that $\varepsilon = \tau(s)(s \m t) - \alpha(s)(\m)\varphi(s, \m)(t) \gr 0$. (The case $\varepsilon \ls 0$ is similar).
		By \eqref{eq:subseqmdp}, there exists $m_1$ such that for all $m \gr m_1$ we have $|\alpha_{i_{m}}(s) -\alpha(s)| \ls \frac{\varepsilon}{2}$. 
		There also exists $m_2$ such that for all $m \gr m_2$ we have $|\tau_{i_{m}}(s) - \tau(s)| \ls \frac{\varepsilon}{2}$. Let $m' = \max\{m_1, m_2\}$ and $m \gr m'$.
		We have the following contradiction:
		\begin{align*}
			& \; \varepsilon \\
			= & \; \tau(s)(s \m t) - \alpha(s)(\m)\varphi(s, \m)(t) \\
			= & \; \tau(s)(s \m t) - \alpha(s)(\m)\varphi(s, \m)(t) - \tau_{i_m}(s)(s \m t) + \alpha_{i_m}(s)(\m)\varphi(s, \m)(t) \\
			&\commenteq{In the induced LMC $\D(\alpha_{i_m})$, we have $\tau_{i_m}(s)(s \m t) = \alpha_{i_m}(s)(\m)\varphi(s, \m)(t)$}\\
			= & \; \big( \tau(s)(s \m t) - \tau_{i_m}(s)(s \m t) \big) + \varphi(s, \m)(t)  \big(\alpha_{i_m}(s)(\m) - \alpha(s)(\m) \big) \\
			\le & \; |\tau(s) - \tau_{i_m}(s)| +  \varphi(s, \m)(t)|\alpha_{i_m}(s) - \alpha(s)| \\
			\ls & \; \frac{\varepsilon}{2} +  \varphi(s, \m)(t) \cdot \frac{\varepsilon}{2} \\
			\le & \; \frac{\varepsilon}{2}  + \frac{\varepsilon}{2} \\
			= & \; \varepsilon \qedhere
		\end{align*}
	\end{proof}
	
	\begin{lemma}\label{lemma:MDP-game-step2}
		Let $\D = <S, \Act, L, \varphi, \ell>$ be an MDP.
		Let $\alpha_1, \alpha_2, \ldots$ be an infinite sequence of strategies and $S_0 \subseteq S$ be a finite subset with $\liminf_{i \to \infty} \max \{d_{\D(\alpha_i)} (s, t) \mid s, t \in S_0)\} = 0$.
		Then there are
		\begin{itemize}
			\item a memoryless strategy $\alpha$ on $S_0$; 
			\item a set $S' \subseteq 2^{S}$ such that for any $E \in S'$ and any $s, t \in E$ we have $\ell(s) = \ell(t)$;
			\item a distribution $\upsilon \in \Dist(S')$;
			\item a function $f: S_0 \times \Act \times S \to S'$ with $t \in f(s, \m, t)$ for all $ (s, \m, t) \in S_0 \times \Act \times S$ such that for all $s \in S_0$ and all $E \in S'$ we have
			$\upsilon(E) = \sum_{\m \in \Act(s)} \sum_{t \in S \text{ s.t. } f(s, \m, t) = E}\alpha(s)(\m)\varphi(s, \m)(t)$.
		\end{itemize}
		Furthermore, for each $E \in S'$, there exists a sequence of strategies $\alpha_1', \alpha_2', \ldots$ such that $\liminf_{i \to \infty} \max \{d_{\D(\alpha_i')} (s, t) \mid s, t \in E\} = 0$, i.e., there is a subsequence of $\D(\alpha_1'), \D(\alpha_2'), \ldots$ such that all pairwise distances in~$E$ converge to~$0$.
	\end{lemma}
	
	\begin{proof}
		
		For $\rho \in \Paths(\D)$, define $N(\rho) = \left\{\rho \m t \mid s=\last(\rho), \m \in \Act(s), t \in \support(\varphi(s)(\m)) \right\}$.
		Write $N := \bigcup_{s \in S_0} N(s)$.
		
		Let $\tau_i$ be the probability transition function for the induced LMC $\D(\alpha_i)$. 
		By \cref{lemma:MDP-game-step}, there are a strategy $\alpha: S_0 \to \Dist(\Act)$, a transition function $\tau: S_0 \to \Dist(\Paths(\D))$, a partition $\{S_1, \ldots, S_k\}$ of $N$ and
		a probability distribution $\mu$ on $\{1, \ldots, k\}$ such that for all $s \in S_0$ and $s\m t \in N$
		\begin{equation}\label{eq:transition-action}
			\tau(s)(s\m t) = \alpha(s)(\m)\varphi(s,\m)(t),
		\end{equation}
		for all $s \in S_0$ and $j \in \{1, \ldots, k\}$
		\begin{equation}\label{eq:transition-distribution}
			\tau(s)(S_j) = \mu(j),
		\end{equation}
		and
		\begin{align}\label{eq:subseq-mdp}
			\liminf_{i \to \infty} \max & \left\{ |\alpha_i(s) - \alpha(s)| \mid s \in S_0 \right\} \nonumber\\
			& \; \mbox{} \cup \left\{ |\tau_i(s) - \tau(s)| \mid s \in S_0 \right\} \nonumber\\
			& \; \mbox{} \cup \left\{ |\tau_i(s)(S_j) - \mu(j) | \mid s \in S_0,\ 1 \le j \le k \right\} \nonumber\\
			& \; \mbox{} \cup \left\{ d_{\D(\alpha_i)}(\rho_1,\rho_2) \mid 1 \le j \le k,\ \rho_1,\rho_2 \in S_j \right\} \quad = \quad 0\,;
		\end{align}
		i.e., there is a subsequence of $\D(\alpha_1), \D(\alpha_2), \ldots$ such that for all $s \in S_0$ and all $1 \le j \le k$ the strategy converges to $\alpha$, the transition function converges to~$\tau$, the transition probability into~$S_j$ converges to~$\mu(j)$, and for each $S_j$ all pairwise distances in~$S_j$ converge to~$0$.
		
		We have the memoryless strategy $\alpha$ on $S_0$.
		We define a set $S' \subseteq 2^S$, a distribution $\upsilon \in \Dist(S')$ and a function $f: S_0 \times \Act \times S \to S'$ as follows.
		
		\begin{itemize}
			\item
			For each $j \in \{1, \ldots, k\}$, we define 
			$E_j := \left\{\last(\rho) \mid \rho \in S_j\right\} \subseteq S$.
			Let $S'$ be the set of $E_j$, i.e., $S' := \{ E_1, \ldots,  E_k \} \subseteq 2^{S}$.
			
			By \eqref{eq:subseq-mdp}, there is a subsequence of $\D(\alpha_1), \D(\alpha_2), \ldots$ such that for each $S_j$ with $1 \le j \le k$ all pairwise distances in~$S_j$ converge to~$0$.
			It follows that all states in~$S_j$ have the same label and all states in $E_j$ have the same label.
			\item
			We define $\upsilon \in \Dist(S')$ such that for $E \in S'$ we have $\upsilon(E):= \sum_{j \in \{1 \ldots k\}: E_j = E}\mu(j)$.
			\item
			Since $\{S_1, \ldots, S_k\}$ is a partition of $N$, for each $s \m t \in N$, there is a unique $S_j$ where $s \m t$ belongs to.
			We define a function $f: S_0 \times \Act \times S \to S'$ such that $f(s, \m, t) := E_j$ for each $s \in S_0$, $\m \in \Act$ and $t \in S$  where $S_j$ is the unique set such that $s \m t \in S_j$.
			It is easy to see that $t \in f(s, \m, t)$ for all $ (s, \m, t) \in S_0 \times \Act \times S$.
			
			Next, we show that for all $s \in S_0$ and all $E \in S'$ we have
			$\upsilon(E) = \sum_{\m \in \Act(s)} \sum_{t \in S \text{ s.t. } f(s, \m, t) = E}\alpha(s)(\m)\varphi(s, \m)(t)$.
			For any $E \in S'$, we have
			\begin{align*}
				& \; \upsilon(E)  \\
				= & \; \sum_{j \in \{1 \ldots k\}: E_j = E}\mu(j) \commenteq{by definition of $\upsilon$}\\
				= & \; \sum_{j \in \{1 \ldots k\}: E_j = E}\tau(s)(S_j) \commenteq{by \eqref{eq:transition-distribution}}\\
				= & \; \sum_{j \in \{1 \ldots k\}: E_j = E}\;\;\sum_{s\m t \in S_j}\tau(s)(s\m t)\\
				= & \; \sum_{j \in \{1 \ldots k\}: E_j = E}\;\;\sum_{s\m t \in S_j}\alpha(s)(\m)\varphi(s, \m)(t) \commenteq{by \eqref{eq:transition-action}}\\
				= & \; \sum_{j \in \{1 \ldots k\}: E_j = E} \;\; \sum_{f(s, \m, t) = E_j}\alpha(s)(\m)\varphi(s, \m)(t) \commenteq{by definition of $f$}\\
				= & \; \sum_{f(s, \m, t) = E}\alpha(s)(\m)\varphi(s, \m)(t) \\
				= & \; \sum_{\m \in \Act(s)} \;\;\;\; \sum_{t \in S \text{ s.t. } f(s, \m, t) = E}\alpha(s)(\m)\varphi(s, \m)(t)
			\end{align*}
		\end{itemize}

		
		Finally, for any $E \in S'$, we define a sequence of strategies $\alpha_1', \alpha_2', \ldots$ such that all pairwise distances in $E$ converge to $0$.
		
		Let $E \in S'$.
		We fix a $j \in \{1, \ldots, k\}$ such that $E_j = E$.
		Let $t \in E_j$.
		There are $s \in S_0$ and $\m \in \Act$ such that $s \m t \in S_j$.
		For all $\rho \in \Paths(\D)$ whose first state is $t$, we define $\alpha_i'(\rho) := \alpha_i(s \m \rho) $.
		
		It follows from the definition of $\alpha_i'$ that each state $t \in E$ in the induced LMC $\D(\alpha_i')$ is probabilistically bisimilar with a state $s \m t \in S_j$ in the induced LMC $\D(\alpha_i)$.
		By \eqref{eq:subseq-mdp}, we have that $\liminf_{i \to \infty} \max \left\{ d_{\D(\alpha_i')}(t, t') \mid t, t' \in E \right\} =  0$, i.e., there is a subsequence of $\D(\alpha_1'), \D(\alpha_2'), \ldots$ such that all pairwise distances in~$E$ converge to~$0$.
		This concludes the proof.
	\end{proof}
	
	Assume that there are an infinite sequence of strategies $\alpha_1, \alpha_2, \ldots$ and $S_0 \subseteq S$ with $\liminf_{i \to \infty} \max \{d_{\D(\alpha_i)} (s, t) \mid s, t \in S_0)\} = 0$.
	\cref{lemma:MDP-game-step2} states that Defender has a winning strategy for the attacker-defender game defined from $\D$ and $S_0$ \cite[Section~3.1]{KT2022}.
	It follows from \cite[Proposition~3]{KT2022} that:
	\begin{corollary}\label{cor:limitequalisable-equalisable}
		Let $\D = <S, \Act, L, \varphi, \ell>$ be an MDP.
		Let $\alpha_1, \alpha_2, \ldots$ be an infinite sequence of strategies, and let $S_0 \subseteq S$ be a set of states with $\liminf_{i \to \infty} \max \{d_{\D(\alpha_i)} (s, t) \mid s, t \in S_0)\} = 0$.
		Then there exists a general strategy $\alpha$ for $\D$ such that in the LMC induced by $\alpha$ all states in $S_0$ are probabilistic bisimilar.
	\end{corollary}

	\begin{theorem}\label{thm:mdp-distance-ls1-nes}
		Let $\D = <S, \Act, L, \varphi, \ell>$ be an MDP.
		Let $\alpha$ be a strategy and $s, t \in S$ with $d_{\D(\alpha)}(s, t) \ls 1$.
		There exist a policy $T$ for the LMC $\D(\alpha)$, two states $u, v \in S$, two paths $\rho_1,\rho_2 \in \Paths(\D)$ with $u = \last(\rho_1)$ and $v = \last(\rho_2)$, and a strategy $\alpha'$ such that $\R_{\D(\alpha)}^T((s,t),\{(\rho_1,\rho_2)\}) \gr 0$ and $u$ and $v$ are probabilistically bisimilar in the LMC induced by $\alpha'$.
	\end{theorem}
	\begin{proof}
		Since $d_{\D(\alpha)}(s, t) \ls 1$, by \cref{corollary:smaller-and-smaller}, there exist a policy $T$ for the LMC $\D(\alpha)$ such that for all $\varepsilon \gr 0$ there are $\rho_1, \rho_2\in \Paths(\D)$ with $d_{\D(\alpha)}(\rho_1, \rho_2) \le \varepsilon$ and $\R_{\D(\alpha)}^T((s,t),\{(\rho_1, \rho_2)\}) \gr~0$.
		Therefore, there exist $\varepsilon_1, \varepsilon_2, \ldots$ where $\varepsilon_i \gr 0$ for all $i$ and $\lim_{i \to \infty} \varepsilon_i = 0$, and an infinite sequence of pairs of states of the LMC $\D(\alpha)$ (paths of the MDP $\D$) $(\rho_1^1, \rho_2^1), (\rho_1^2, \rho_2^2), \ldots$ such that $d_{\D(\alpha)}(\rho_1^i, \rho_2^i) \le \varepsilon_i$ and $\R_{\D(\alpha)}^T((s,t),\{(\rho_1^i, \rho_2^i)\}) \gr~0$.

		Let $u^i = \last(\rho_1^i)$ and $v^i = \last(\rho_2^i)$.
		Without loss of generality, assume $u^i \neq v^i$.
		We define the strategy $\alpha_i(u^i \rho) := \alpha(\rho_1^i \rho)$, $\alpha_i(v^i \rho) = \alpha(\rho_2^i \rho)$ for any $\rho \in \Paths(\D)$.
		In the case that $u^i = v^i$, we simply define  $\alpha_i(u^i \rho) = \alpha(\rho_1^i \rho)$ for any $\rho \in \Paths(\D)$.
		By the definition of $\alpha_i$, the state $u^i$ of the LMC $\D(\alpha_i)$ is probabilistically bisimilar with the state $\rho_1^i$ of the LMC $\D(\alpha)$.
		Similarly, the state $v^i$ of the LMC $\D(\alpha_i)$ is probabilistically bisimilar with the state $\rho_2^i$ of the LMC $\D(\alpha)$.
		
		Since $d_{\D(\alpha)}(\rho_1^i, \rho_2^i) \le \varepsilon_i$ for all $i$, we have $d_{\D(\alpha_i)}(u^i, v^i) \le \varepsilon_i$ for all $i$.
		By the pigeonhole principle, there exist $u, v \in S$ such that $u = u^i$ and $v = v^i$ for infinitely many $i$.
		
		Furthermore, there exists a subsequence of $\D(\alpha_1), \D(\alpha_2), \ldots$ such that the distance of $u$ and $v$ converges to $0$, i.e., $\liminf_{i \to \infty} \max\{d_{\D(\alpha_i)}(u, v)\} = 0$.
		
		Thus, by \cref{cor:limitequalisable-equalisable}, there exists a general strategy $\alpha'$ for $\D$ such that $u$ and $v$ are probabilistically bisimilar in the LMC induced by $\alpha'$.
		
		Let $(\rho_1, \rho_2)$ be a pair of states in the LMC $\D(\alpha)$ such that $(\rho_1, \rho_2) = (\rho_1^i, \rho_2^i)$ for some $i$ with $\last(\rho_1^i) = u$ and $\last(\rho_2^i) = v$.
		We have $\R_{\D(\alpha)}^T((s,t),\{(\rho_1, \rho_2)\}) = \R_{\D(\alpha)}^T((s,t),\{(\rho_1^i, \rho_2^i)\}) \gr~0$.
	\end{proof}
	
	\begin{theorem}\label{thm:mdp-distance-ls1-suf}
		Let $\D = <S, \Act, L, \varphi, \ell>$ be an MDP.
		Let $\alpha$ be a strategy, $s, t \in S$ and $\rho_1,\rho_2 \in \Paths(\D)$ such that there is a policy $T$ for the LMC $\D(\alpha)$ with $\R_{\D(\alpha)}^T((s,t),\{(\rho_1,\rho_2)\}) \gr 0$.
		Let $u = \last(\rho_1)$, $v = \last(\rho_2)$.
		If there is a strategy $\alpha'$ such that $d_{\D(\alpha')}(u, v) = 0$,
		there is a strategy $\alpha''$ such that $d_{\D(\alpha'')}(s, t) \ls 1$.
	\end{theorem}
	\begin{proof}
		We define the strategy $\alpha''$ as follows:
		\begin{itemize}
			\item $\alpha''(\rho_1) := \alpha'(u)$;
			\item $\alpha''(\rho_1 \m \rho) := \alpha'(u \m \rho)$ where $\m \in \Act(u)$ and $\rho \in \Paths(\D)$;
			\item $\alpha''(\rho_2) := \alpha'(v)$;
			\item $\alpha''(\rho_2 \m \rho) := \alpha'(v \m \rho)$ where $\m \in \Act(v)$ and $\rho \in \Paths(\D)$;
			\item $\alpha''(\rho) := \alpha(\rho)$ for all other $\rho \in \Paths(\D)$.
		\end{itemize}
		By the definition of $\alpha''$, the state $u$ in the LMC $\D(\alpha')$ is probabilistically bisimilar with the state $\rho_1$ in the LMC $\D(\alpha'')$.
		Similarly, the state $v$ in the LMC $\D(\alpha')$ is probabilistically bisimilar with the state $\rho_2$ in the LMC $\D(\alpha'')$.
		Since $d_{\D(\alpha')}(u, v) = 0$, we have $d_{\D(\alpha'')}(\rho_1, \rho_2) = 0$.
		
		We denote by $\tau$ the transition function of the LMC $\D(\alpha'')$.
		Let $T'$ be a policy for the LMC $\D(\alpha'')$ such that
		\begin{itemize}
			\item $T'(\rho_1', \rho_2') := T(\rho_1', \rho_2')$ where $\rho_1'$ is a prefix of $\rho_1$ and $\rho_2'$ is a prefix of $\rho_2$;
			\item $T'(\rho, \rho') := \omega$ where $\omega \in \Omega(\tau(\rho), \tau(\rho'))$ for all other state  $\rho, \rho'$ in $\D(\alpha'')$.
		\end{itemize}
		
		By definition of $T'$, we have $\R_{\D(\alpha'')}^{T'}((s,t),\{(\rho_1,\rho_2)\}) =\R_{\D(\alpha)}^T((s,t),\{(\rho_1,\rho_2)\}) \gr 0$.
		Together with $d_{\D(\alpha'')}(\rho_1, \rho_2) = 0$ and \cref{corollary:lmc-distance-graph-reachability}, we have $d_{\D(\alpha'')}(s, t) \ls 1$.
	\end{proof}
	
	\Cref{thm:mdp-distance-ls1-nes,thm:mdp-distance-ls1-suf} together imply \cref{prop:mdp-distance-ls1} from the main body.
	
	\subsection{Proof of \texorpdfstring{\cref{thm:distance-ls-one-exptimehard}}{Theorem \ref{thm:distance-ls-one-exptimehard}}} \label{proof:thm:distance-ls-one-exptimehard}
	
	

	
	Let $\D = <S, \Act, L, \varphi, \ell>$ be an MDP and $s_1, s_2$ be two states.
	We construct an MDP $\D'$ from $\D$.
	The MDP $\D'$ consists of two disjoint parts: $\D_1$ and $\D_2$; see \cref{fig:MDP-EXPTIME-Hard}.
	Let $i \in \{1, 2\}$.
	The set of states of $\D_i$ is $\{(s, i) \mid s \in S \} \cup \{\$_i\}$ where $\$_i$ is a fresh state.
	The transitions $\varphi_i$ are defined as follows:
	\begin{itemize}
		\item
		The state $(s, i)$ has all the actions of $s$;
		the action $\m' \in \Act(s)$ goes with probability $\frac{1}{2} \varphi(s, \m')(t)$ to $(t, i)$, that is, $\varphi_i\big( (s, i), \m' \big)\big( (t, i) \big) = \frac{1}{2} \varphi(s, \m')(t)$ for all $s,t \in S$ and $\m' \in \Act(s)$;
		the action $\m' \in \Act(s)$ goes with probability $\frac{1}{2}$ to $\$_i$, that is, $\varphi_i\big( (s, i), \m' \big) (\$_i)  = \frac{1}{2}$.
		\item
		The state $\$_i$ transitions with probability $1$ to $(s_i, i)$, via a default action~$\m$.
	\end{itemize}
	Each state $(s, i)$ is labelled with $\ell(s)$.
	The state $\$_i$ is labelled with $\$$ where $\$$ is a fresh label.
	\begin{figure}[!htb]	
		\centering
		\tikzstyle{BoxStyle} = [draw, circle, fill=black, scale=0.4,minimum width = 1pt, minimum height = 1pt]
		
		\begin{tikzpicture}[xscale=.6,>=latex',shorten >=1pt,node distance=3cm,on grid,auto]
			\node[state] (s) at (0, 0){$s$};
			\node[BoxStyle] (ds) at (2, 0){};
			\node[label] at (2, 0.3) {$\m'$};
			\node[state, fill=orange!20] (t) at (5, 0){$t$};
			
			\node[label] at (2,-2) {\parbox{4cm}{$p$ is transition probability: $p = \varphi(s, \m')(t)$.}};
			
			\path[-] (s) edge node [midway, above] {} (ds);
			\path[->] (ds) edge node [midway, above] {$p$} (t);
			
			\node[state, fill=red!20] (s1) at (0, -3.5) {$s_1$};
			\node[state, fill=red!20] (s2) at (5, -3.5) {$s_2$};
			\node[label] at (2,-4.5) {\parbox{4cm}{$s_1, s_2$ are two states in $\D$.}};

			\node[state] (si) at (10, 0) {$(s, i)$};
			\node[BoxStyle] (dsi) at (12, 0){};
			\node[label] at (12, 0.3) {$\m'$};
			\node[state, fill=orange!20] (ti) at (15, 0){$(t, i)$};
			\node[state, fill=blue!20] (dollar) at (12, -1.5) {$\$_i$};
			\node[state, fill=red!20] (sii) at (12, -3.5) {$(s_i, i)$};
			
			\path[-] (si) edge node {} (dsi);
			\path[->] (dsi) edge node {$\frac{p}{2}$} (ti);
			\path[->] (dsi) edge node {$\frac{1}{2}$} (dollar);
			\path[->] (dollar) edge node {$1$} (sii);
			
			\node[label] at (2.5, 1.2) {The MDP $\D = <S, \Act, L, \varphi, \ell>$};
			\node[rectangle,draw,dashed, minimum width = 5cm, minimum height = 5.5cm] at (2.5,-2) {};
			\node[label] at (13, 1.2) {The MDP $\D_i$ where $i \in \{1, 2\}$};
			\node[rectangle,draw,dashed, minimum width = 6cm, minimum height = 5.5cm] at (13,-2) {};
		\end{tikzpicture}
		
		\caption{
			The state $(s, i)$ in the MDP $\D_i$, where $s$ is a state of $\D$, has the same actions as the state $s$.
			For state $(s, i)$, the action $\m' \in \Act(s)$ transitions to $(t, i)$ with probability $\frac{1}{2}\varphi(s, \m')(t)$ and $\$_i$ with probability $\frac{1}{2}$.
			The new state $\$_i$, labelled with a new label $\$$, transitions to $(s_i, i)$ with probability one.
		}
		\label{fig:MDP-EXPTIME-Hard}
	\end{figure}
	We show that there is a general strategy $\alpha$ for $\D$ such that $d_{\D(\alpha)}(s_1, s_2) = 0$ if and only if there is a general strategy $\alpha'$ for $\D'$ such that $d_{\D'(\alpha')}((s_1,1), (s_2,2)) \ls 1$.

	For each path $u_1\m_1u_2\ldots$ of $\D$, we have a path $(u_1, 1)\m_1(u_2, 1) \ldots $ in the MDP $\D_1$ and a path $(u_1, 2)\m_1(u_2, 2) \ldots $ in the MDP $\D_2$, respectively.
	We define a function $p_1$ which maps a path in $\D$ to such a path in $\D_1$.
	Similarly, we define a function $p_2$ which maps a path in $\D$ to its corresponding path in $\D_2$.
	
	We define a function $p': \Paths(\D') \to \Paths(\D) \cup \{\$\}$. For a path $\rho' \in \Paths(\D')$, we have
	\[
	p'(\rho') := \left \{
	\begin{array}{ll}
		\rho & \mbox{if $\rho \in \Paths(\D)$ such that $p_1(\rho)=\rho'$ or $p_2(\rho) = \rho'$;}\\
		\rho & \mbox{if $\$_1\m\rho''$ is a suffix of $\rho'$ and $\rho \in \Paths(\D)$ such that $p_1(\rho) = \rho''$;}\\
		\rho & \mbox{if $\$_2\m\rho''$ is a suffix of $\rho'$ and $\rho \in \Paths(\D)$ such that $p_2(\rho) = \rho''$;}\\
		\$ & \mbox{otherwise, i.e., $\rho'$ ends with $\$_1$ or $\$_2$.}
	\end{array}
	\right .
	\]
	
	We show that there is a general strategy $\alpha$ for $\D$ such that $d_{\D(\alpha)}(s_1, s_2) = 0$ if and only if there is a general strategy $\alpha'$ for $\D'$ such that $d_{\D'(\alpha')}((s_1,1), (s_2,2)) \ls 1$.
	
	\begin{itemize}
		\item ``$\implies$'':
		Assume there is a general strategy $\alpha$ for $\D$ such that $d_{\D(\alpha)}(s_1, s_2) = 0$.
		We show that there exists a general strategy $\alpha'$ for $\D'$ such that $d_{\D'(\alpha')}\big((s_1, 1), (s_2,2) \big) = 0 \ls 1$.
		
		We define the strategy $\alpha'$ for $\D'$ as follows:
		\begin{itemize}
			\item
			For a path $\rho'$ in $\D'$ such that there is a path $\rho$ in $\D$ with $\rho = p'(\rho')$, play the strategy $\alpha(\rho)$, that is, $\alpha'(\rho') := \alpha(\rho)$;
			\item
			For a path $\rho \in \D'$ that ends with $\$_1$ or $\$_2$, since there is only one default action $\m$, play $\m$ with probability one, that is, $\alpha'(\rho) = \mathbf{1}_{\m}$.
		\end{itemize}
		
		Let $\tau$ and $\tau'$ be the probability transition functions for the induced LMCs $\D(\alpha)$ and $\D'(\alpha')$, respectively.
		
		To prove that $(s_1, 1)$ and $(s_2, 2)$ are probabilistically bisimilar (have distance zero) in the LMC $\D'(\alpha')$, we define a relation $R$ for which $\big((s_1, 1), (s_2, 2)\big) \in R$ and show that $R$ is a probabilistic bisimulation.
		
		We define the relation as follows.
		Two states $\rho_1$ and $\rho_2$ in $\D'(\alpha')$ are related in $R$ if and only if either $p'(\rho_1) = p'(\rho_2) = \$$ or $p'(\rho_1)$ and $p'(\rho_2)$ are probabilistically bisimilar in $\D(\alpha)$.

		
		It holds that $(s_1, 1)$ and $(s_2, 2)$ are related in $R$ since $p'\big( (s_1,1) \big) = s_1$, $p'\big( (s_2,2) \big) = s_2$ and $s_1 \sim_{\D(\alpha)} s_2$.
		
		Next, we show that $R$ is a probabilistic bisimulation.
		It is easy to see that $R$ is an equivalence relation as it is reflexive, symmetric and transitive.
		To prove $R$ is a probabilistic bisimulation, we show that for all $\rho_1, \rho_2 \in R$ it holds that $\rho_1$ and $\rho_2$ have the same label and $\tau'(\rho_1)(E') = \tau'(\rho_2)(E')$ for each $R$-equivalence class $E'$.
		
		Let $(\rho_1, \rho_2) \in R$.
		We distinguish the following cases:
		\begin{itemize}
			\item
			Assume $\rho_1$ and $\rho_2$ are labelled with $\$$.
			Let $E'$ be the $R$-equivalence class that contains those $\rho \in \Paths(\D')$ that end with either $\$_1\m(s_1,1)$ or $\$_2\m(s_2,2)$.
			We have $\tau'(\rho_1)(E') = 1 = \tau'(\rho_2)(E')$.
			\item
			Assume neither $\rho_1$ nor $\rho_2$ is labelled with $\$$.
			Let $\rho = p'(\rho_1)$ and $\rho' = p'(\rho_2)$.
			By definition of $R$, we have $\rho \sim_{\D(\alpha)} \rho'$.
			Thus, $\rho_1$ and $\rho_2$ have the same label.
			
			We consider the $R$-equivalence classes.
			Suppose $E'$ is an $R$-equivalence class which contains all states (paths in $\D'$) in $\D'(\alpha')$ that end with either $\$_1$ or $\$_2$.
			We have $\tau'(\rho_1)(E') = \frac{1}{2} = \tau'(\rho_2)(E')$.
			Now suppose $E'$ is an $R$-equivalence class which contains $\rho_1'$, a successor of $\rho_1$ in $\D'(\alpha')$ such that $\last(\rho_1') \neq \$_1$ and $\last(\rho_1') \neq \$_2$.
			There is a $\sim_{\D(\alpha)}$-equivalence class $E$ which contains $p'(\rho_1')$.
			Since $\rho \sim_{\D(\alpha)} \rho'$, we have $\tau(\rho)(E) = \tau(\rho')(E)$.
			Furthermore, for any $\rho_2' \in E'$ that is a successor of $\rho_2$, we have $p'(\rho_2') \in E$ and $p'(\rho_2')$ is a successor of $\rho'$.
			Together with $\tau'(\rho_1)(E') = \frac{1}{2}\tau(\rho)(E)$ and $\tau'(\rho_2)(E') = \frac{1}{2}\tau(\rho')(E)$,
			we have $\tau'(\rho_1)(E') = \tau'(\rho_2)(E')$.
		\end{itemize}
		
		\item ``$\impliedby$'':
		Assume there is a general strategy $\alpha'$ for $\D'$ such that $d_{\D'(\alpha')}\big( (s_1,1), (s_2,2) \big) \ls 1$.
		We show that there is a general strategy $\alpha$ for $\D$ such that $d_{\D(\alpha)}(s_1, s_2) = 0$.

		By \cref{thm:mdp-distance-ls1-nes}, there exist a policy $T$ for the LMC $\D'(\alpha')$, two states $u, v$, two paths $\rho_1, \rho_2 \in \Paths(\D')$ with $u = \last(\rho_1)$ and $v = \last(\rho_2)$, and a strategy $\alpha''$ such that $\R_{\D'(\alpha')}^T( ( (s_1,1), (s_2,2)),\{(\rho_1,\rho_2)\}) \gr 0$ and $u$ and $v$ are probabilistically bisimilar in the LMC induced by $\alpha''$.
		
		We have that $u$ and $v$ have the same label.
		Furthermore, $u$ is in $\D_1$ and $v$ is in $\D_2$.
		We distinguish the following two cases:
		\begin{itemize}
			\item
			Assume that $u = \$_1$ and $v = \$_2$.
			We define a strategy $\alpha$ for $\D$:
			for all $\rho \in \Paths(\D)$ such that $\rho$ begins with $s_1$ we have $\alpha(\rho) := \alpha''(\$_1\m p_1(\rho))$;
			for all $\rho \in \Paths(\D)$ such that $\rho$ begins with $s_2$ we have $\alpha(\rho) := \alpha''(\$_2 \m p_2(\rho))$.
			
			Let $\tau$ be the probability transition function for the induced LMC $\D(\alpha)$.
			To show that $s_1$ and $s_2$ are probabilistically bisimilar in  $\D(\alpha)$, we define a relation $R$ such that $(s_1, s_2) \in R$ and show that $R$ is a probabilistic bisimulation.
			
			A pair of states $(\rho, \rho')$ in $\D(\alpha)$ are in $R$ if and only if one of the following holds:
			\begin{itemize}
				\item
				Both $\rho$ and $\rho'$ begin with $s_1$ and $\$_1\m p_1(\rho) \sim_{\D'(\alpha'')} \$_1\m p_1(\rho')$;
				\item
				Both $\rho$ and $\rho'$ begin with $s_2$ and $\$_2\m p_2(\rho) \sim_{\D'(\alpha'')} \$_2\m p_2(\rho')$;
				\item
				The state $\rho$ begins with $s_1$, $\rho'$ begins with $s_2$ and $\$_1\m p_1(\rho) \sim_{\D'(\alpha'')} \$_2\m p_2(\rho')$;
				\item
				The state $\rho$ begins with $s_2$, $\rho'$ begins with $s_1$ and $\$_2\m p_2(\rho) \sim_{\D'(\alpha'')} \$_1\m p_1(\rho')$;
				\item
				The state $\rho$ or $\rho'$ does not begin with $s_1$ or $s_2$.
			\end{itemize}
			
			It is not hard to see that $R$ is an equivalence relation since it is reflexive, symmetric and transitive.
			To prove that it is a probabilistic bisimulation, similar to the previous direction, we can show that for all $\rho, \rho' \in R$ it holds that $\rho$ and $\rho'$ have the same label and $\tau(\rho)(E) = \tau(\rho')(E)$ for each $R$-equivalence class $E$.
			
			\item
			Assume that $u = (s, 1)$ and $v = (t, 2)$ with $\ell(s) = \ell(t)$.
			Let $(s, 1) \m_1 \$_1$ and $(t,1) \m_2 \$_2$ be successors of $(s, 1)$ and $(t, 2)$ in $\D'(\alpha'')$, respectively.
			We define a strategy $\alpha$ for $\D$:
			for all $\rho \in \Paths(\D)$ such that $\rho$ begins with $s_1$ we have $\alpha(\rho) := \alpha''((s, 1) \m_1 \$_1 \m p_1(\rho))$;
			for all $\rho \in \Paths(\D)$ such that $\rho$ begins with $s_2$ we have $\alpha(\rho) := \alpha''((t, 2) \m_2 \$_2 \m p_2(\rho))$;
			
			To show that $s_1$ and $s_2$ are probabilistically bisimilar in  $\D(\alpha)$, we define a relation $R$ such that $(s_1, s_2) \in R$ and show that $R$ is a probabilistic bisimulation.
			
			A pair of states $(\rho, \rho')$ in $\D(\alpha)$ are in $R$ if and only if one of the following holds:
			\begin{itemize}
				\item
				Both $\rho$ and $\rho'$ begin with $s_1$ and $(s, 1) \m_1 \$_1 \m p_1(\rho) \sim_{\D'(\alpha'')} (s, 1) \m_1 \$_1 \m p_1(\rho')$;
				\item
				Both $\rho$ and $\rho'$ begin with $s_2$ and $(t, 2) \m_2 \$_2 \m p_2(\rho) \sim_{\D'(\alpha'')} (t, 2) \m_2 \$_2 \m p_2(\rho')$;
				\item
				The state $\rho$ begins with $s_1$, $\rho'$ begins with $s_2$ and $(s, 1) \m_1 \$_1 \m p_1(\rho) \sim_{\D'(\alpha'')} (t, 2) \m_2 \$_2 \m p_2(\rho')$;
				\item
				The state $\rho$ begins with $s_2$, $\rho'$ begins with $s_1$ and $(t, 2) \m_2 \$_2 \m p_2(\rho') \sim_{\D'(\alpha'')} (s, 1) \m_1 \$_1 \m p_1(\rho)$;
				\item
				The state $\rho$ or $\rho'$ does not begin with $s_1$ or $s_2$.
			\end{itemize}
			It is not hard to see that $R$ is an equivalence relation since it is reflexive, symmetric and transitive.
			To prove that it is a probabilistic bisimulation, we can show that for all $\rho, \rho' \in R$ it holds that $\rho$ and $\rho'$ have the same label and $\tau(\rho)(E) = \tau(\rho')(E)$ for each $R$-equivalence class $E$ where $\tau$ is the probability transition function for $\D(\alpha)$.
		\end{itemize}

		
		
			
			
		
	\end{itemize}
	This concludes the proof of \cref{thm:distance-ls-one-exptimehard}.
	
\end{document}

%% file: noninterference.tex
\section{Probabilistic Noninterference}
	\label{section:noninterference}
	\QT{In this section we provide examples that show some challenges in distance minimisation and illustrate the relation between distance minimisation and probabilistic noninterference in security. As described in the introduction, we are interested in schedulers that minimise the information leakage.
	}

\begin{example} \label{ex:non1}
    We borrow an example from \cite[Section~4]{SabelfeldS00} and \cite[Section~3]{KT2022}.
    Consider the following simple program composed of two threads, involving a \emph{high} boolean variable~$h$ (high confidentiality) and a \emph{low} boolean variable~$l$ (observable):
    \[
     l \ := \ h \quad \vert \quad l \ := \ \neg h
    \]
    The vertical bar $\mathord{\vert}$ separates two threads.
    The order in which the threads are executed is determined by a scheduler.
    We assume that assignments to the value of variable~$l$ are visible.
    One may model the program as the following MDP in \cref{fig:security-example1}.
    \begin{figure}[!htb]
	\centering
    	\begin{tikzpicture}[yscale=0.6,xscale=.6,>=latex',shorten >=1pt,node distance=3cm,on grid,auto]
		\node[state] (s0) at (0,0){$s_0$};
		\node[BoxStyle] (m0) at (-2.5,-0.5) {};
		\node[label] at (-3.2, -0.5) {$\m_0$};
		\node[state, fill=red!20] (h0m0l0) at (-2.5,-2){\color{blue}$0$};
		\node[state] (left0) at (-2.5, -4) {};
		\node[state, fill=OliveGreen!30] (h0m0l0ml1) at (-2.5,-6){\color{blue}$1$};
		
		\node[BoxStyle] (m1) at (2.5,-0.5) {};
		\node[label] at (3.2, -0.5) {$\m_1$};
		\node[state, fill=OliveGreen!30] (h0m1l1) at (2.5,-2){\color{blue}$1$};    
		\node[state] (right0) at (2.5, -4) {};
		\node[state, fill=red!20] (h0m1l1ml0) at (2.5,-6){\color{blue}$0$};
		
		\node[state] (t0) at (0,-8) {$t_0$};
		
		\path[-] (s0) edge node {} (m0);
		\path[-] (s0) edge node {} (m1);
		\path[->] (m0) edge node {} (h0m0l0);
		\path[->] (m1) edge node {} (h0m1l1);
		\path[->] (h0m0l0) edge node {} (left0);
		\path[->] (left0) edge node {} (h0m0l0ml1);
		\path[->] (h0m1l1) edge node {} (right0);
		\path[->] (right0) edge node {} (h0m1l1ml0);
		\path[->] (h0m0l0ml1) edge (t0);
		\path[->] (h0m1l1ml0) edge (t0);
        \path (t0) edge [loop right] (t0);
		
		\node[state] (s0) at (11,0){$s_1$};
		\node[BoxStyle] (m0) at (8.5,-0.5) {};
		\node[label] at (7.8, -0.5) {$\m_0$};
		\node[state, fill=OliveGreen!30] (h0m0l0) at (8.5,-2){\color{blue}$1$};
		\node[state] (left0) at (8.5, -4) {};
		\node[state, fill=red!20] (h0m0l0ml1) at (8.5,-6){\color{blue}$0$};
		
		\node[BoxStyle] (m1) at (13.5,-0.5) {};
		\node[label] at (14.2, -0.5) {$\m_1$};
		\node[state, fill=red!20] (h0m1l1) at (13.5,-2){\color{blue}$0$};    
		\node[state] (right0) at (13.5, -4) {};
		\node[state, fill=OliveGreen!30] (h0m1l1ml0) at (13.5,-6){\color{blue}$1$};
		
		\node[state] (t0) at (11,-8) {$t_1$};

		\path[-] (s0) edge node {} (m0);
		\path[-] (s0) edge node {} (m1);
		\path[->] (m0) edge node {} (h0m0l0);
		\path[->] (m1) edge node {} (h0m1l1);
		\path[->] (h0m0l0) edge node {} (left0);
		\path[->] (left0) edge node {} (h0m0l0ml1);
		\path[->] (h0m1l1) edge node {} (right0);
		\path[->] (right0) edge node {} (h0m1l1ml0);
		\path[->] (h0m0l0ml1) edge node [midway, left] {} (t0);
		\path[->] (h0m1l1ml0) edge node [midway, right]{} (t0);
        \path (t0) edge [loop right] (t0);
	\end{tikzpicture}
	\caption{
	The program from \cref{ex:non1} as an MDP.
	The states $s_0$ and $s_1$ have two available actions, $\m_0$ and $\m_1$.
	The default action $\m$ for the other states is omitted.
	Different colours (state labels) indicate the distinct values of the low data.
	Throughout the paper, transition probabilities out of each action are one unless explicitly specified.
	}
	\label{fig:security-example1}
	\end{figure}
    Here, $s_0$ and~$s_1$ correspond to initial states with $h=0$ and $h=1$, respectively.
    \QT{
    The two actions in the MDP, $\m_0$ and $\m_1$, correspond to the two possible orders of execution: action $\m_0$
    models the choice of executing $l  := h $ first, followed by $l  :=  \neg h$, while $\m_1$ models the reverse order.
    The different colours represent the distinct values of the low, observable data.
    For instance, in state $s_0$, if the scheduler selects $\m_0$ (the left branch of $s_0$), then $l$ becomes $0$ after
    executing  $l := h $ and $1$ after executing $l := \neg h$.
	All transitions are with probability one.}
    A memoryless strategy that chooses actions $\m_0, \m_1$ uniformly at random (i.e., with probability $0.5$ each) makes $s_0, s_1$ probabilistic bisimilar; i.e., $d(s_0,s_1) = 0$ under this strategy.
\lipicsEnd
\end{example}

\begin{example} \label{ex:non2}
Consider the following variant of \cref{ex:non1}.
\[
\begin{split}
& \text{repeat} \\
& \hspace{5mm} l \ := \ h \quad \vert \quad l \ := \ \neg h \\
& \text{until } \mathit{coin}(p) \lor h
\end{split}
\]
Here, $\mathit{coin}(p)$, for a fixed parameter $p \in [0,1]$, models a biased coin that returns $\mathit{true}$ with probability~$p$ and $\mathit{false}$ with probability~$1-p$.
One may model the program as the MDP in \cref{fig:security-example}, \emph{except} that $t_0, t_1$ are sinks, as in \cref{ex:non1}.
The value of~$h$ influences the termination condition of the loop and therefore ``leaks'' (with probability $1-p$).
As a result, under the optimal (in terms of minimising the distance) strategy, which is the same as in \cref{ex:non1}, we have now $d(s_0,s_1) = 1-p$.
The smaller~$p$, the ``worse'' the leak.
\lipicsEnd
\end{example}

The following example shows that general strategies may be needed for optimal security.

\begin{example} \label{ex:non3}
In order to mitigate the leak from \cref{ex:non2}, one might extend the program as follows, so that the scheduler is given an opportunity to disguise the fact that the program with $h=1$ tends to terminate earlier than the program with $h=0$:
\[
\begin{split}
& \text{repeat} \\
& \hspace{5mm} l \ := \ h \quad \vert \quad l \ := \ \neg h \\
& \text{until } \mathit{coin}(p) \lor h \\
& \text{repeat forever} \\
& \hspace{5mm} l \ := \ 0 \quad \oplus \quad l \ := \ 1
\end{split}
\]
Here, $\mathord{\oplus}$ stands for a nondeterministic choice, to be made by the scheduler, \QT{where exactly one of the
instructions $ l :=  0 $ and $ l := 1 $ is executed.}
\QT{In \cref{fig:security-example},  this corresponds to taking actions $\m_2$ and $\m_3$, respectively.}
	\begin{figure}[!htb]	
	\centering	
	\begin{tikzpicture}[yscale=0.65,xscale=.6,>=latex',shorten >=1pt,node distance=3cm,on grid,auto]
		\node[state] (s0) at (0,0){$s_0$};
		\node[BoxStyle] (m0) at (-2.5,-0.5) {};
		\node[label] at (-3.2, -0.5) {$\m_0$};
		\node[state, fill=red!20] (h0m0l0) at (-2.5,-2){\color{blue}$0$};
		\node[state] (left0) at (-2.5, -4) {};
		\node[state, fill=OliveGreen!30] (h0m0l0ml1) at (-2.5,-6){\color{blue}$1$};
		
		\node[BoxStyle] (m1) at (2.5,-0.5) {};
		\node[label] at (3.2, -0.5) {$\m_1$};
		\node[state, fill=OliveGreen!30] (h0m1l1) at (2.5,-2){\color{blue}$1$};    
		\node[state] (right0) at (2.5, -4) {};
		\node[state, fill=red!20] (h0m1l1ml0) at (2.5,-6){\color{blue}$0$};
		
		\node[state] (t0) at (0,-8) {$t_0$};
		\node[BoxStyle] (m2) at (-2.5,-8.5) {};
		\node[label] at (-3.2, -8.5) {$\m_2$};
		\node[state, fill=red!20] (t0l0) at (-2.5,-10) {\color{blue}$0$};
		\node[BoxStyle] (m3) at (2.5,-8.5) {};
		\node[label] at (3.2, -8.5) {$\m_3$};
		\node[state, fill=OliveGreen!30] (t0l1) at (2.5,-10) {\color{blue}$1$};
		
		\path[-] (s0) edge node {} (m0);
		\path[-] (s0) edge node {} (m1);
		\path[->] (m0) edge node {} (h0m0l0);
		\path[->] (m1) edge node {} (h0m1l1);
		\path[->] (h0m0l0) edge node {} (left0);
		\path[->] (left0) edge node {} (h0m0l0ml1);
		\path[->] (h0m1l1) edge node {} (right0);
		\path[->] (right0) edge node {} (h0m1l1ml0);
		\path[->] (h0m0l0ml1) edge node [yshift=-1mm,midway, left] {$p$} (t0);
		\path[->] (h0m1l1ml0) edge node [yshift=-1mm, midway, right]{$p$} (t0);
		\path[->] (h0m0l0ml1) edge [bend right=20] node [right, near start,xshift=-1em,yshift=-1.2em]{$1-p$} (s0);
		\path[->] (h0m1l1ml0) edge [bend left=20] node [left, near start,xshift=1em,yshift=-1.2em]{$1-p$} (s0);

		\path[-] (t0) edge node {} (m2);
		\path[-] (t0) edge node {} (m3);
		\path[->] (m2) edge node {} (t0l0);
		\path[->] (m3) edge node {} (t0l1);
		\path[->] (t0l0) edge node {} (t0);
		\path[->] (t0l1) edge node {} (t0);
		
		\node[state] (s0) at (11,0){$s_1$};
		\node[BoxStyle] (m0) at (8.5,-0.5) {};
		\node[label] at (7.8, -0.5) {$\m_0$};
		\node[state, fill=OliveGreen!30] (h0m0l0) at (8.5,-2){\color{blue}$1$};
		\node[state] (left0) at (8.5, -4) {};
		\node[state, fill=red!20] (h0m0l0ml1) at (8.5,-6){\color{blue}$0$};
		
		\node[BoxStyle] (m1) at (13.5,-0.5) {};
		\node[label] at (14.2, -0.5) {$\m_1$};
		\node[state, fill=red!20] (h0m1l1) at (13.5,-2){\color{blue}$0$};    
		\node[state] (right0) at (13.5, -4) {};
		\node[state, fill=OliveGreen!30] (h0m1l1ml0) at (13.5,-6){\color{blue}$1$};
		
		\node[state] (t0) at (11,-8) {$t_1$};
		\node[BoxStyle] (m2) at (8.5,-8.5) {};
		\node[label] at (7.8, -8.5) {$\m_2$};
		\node[state, fill=red!20] (t0l0) at (8.5,-10) {\color{blue}$0$};
		\node[BoxStyle] (m3) at (13.5,-8.5) {};
		\node[label] at (14.2, -8.5) {$\m_3$};
		\node[state, fill=OliveGreen!30] (t0l1) at (13.5,-10) {\color{blue}$1$};
		
		\path[-] (s0) edge node {} (m0);
		\path[-] (s0) edge node {} (m1);
		\path[->] (m0) edge node {} (h0m0l0);
		\path[->] (m1) edge node {} (h0m1l1);
		\path[->] (h0m0l0) edge node {} (left0);
		\path[->] (left0) edge node {} (h0m0l0ml1);
		\path[->] (h0m1l1) edge node {} (right0);
		\path[->] (right0) edge node {} (h0m1l1ml0);
		\path[->] (h0m0l0ml1) edge node [midway, left] {} (t0);
		\path[->] (h0m1l1ml0) edge node [midway, right]{} (t0);
		
		\path[-] (t0) edge node {} (m2);
		\path[-] (t0) edge node {} (m3);
		\path[->] (m2) edge node {} (t0l0);
		\path[->] (m3) edge node {} (t0l1);
		\path[->] (t0l0) edge node {} (t0);
		\path[->] (t0l1) edge node {} (t0);
	\end{tikzpicture}
	
	\caption{
		The program from \cref{ex:non3} as an MDP.
		The states $s_0$ and $s_1$ have two available actions, $\m_0$ and $\m_1$.
		The states $t_0$ and $t_1$ also have two available actions, $\m_2$ and $\m_3$.	
		Different colours (state labels) indicate the distinct values of the low data.
	}
	\label{fig:security-example}
\end{figure}
One can show that the optimal \emph{memoryless} strategy chooses between $\m_0$ and~$\m_1$ uniformly at random (as before), and also chooses between $\m_2$ and~$\m_3$ uniformly at random.
Under this strategy we have $d(s_0, t_1) = 0.5 + 0.5(1-p) d(s_0,t_1)$, implying $d(s_0,t_1) = \frac{1}{1+p}$, and thus $d(s_0,s_1) = (1-p) d(s_0,t_1) = \frac{1-p}{1+p}$, which is, for $p \in (0,1)$, smaller (i.e., better) than the distance achievable in \cref{ex:non2}.

However, there is a general strategy~$\alpha$, not memoryless, that perfectly disguises when the first loop is exited.
This strategy~$\alpha$ chooses between $\m_0$ and~$\m_1$ uniformly at random (as before).
When the execution path visits $t_0$ or~$t_1$ for the $i$th time, $i \ge 1$, then, if $i$ is odd, $\alpha$ chooses between $\m_2$ and~$\m_3$ uniformly at random, and if $i$~is even, $\alpha$ chooses the action that was not taken upon the $(i-1)$th visit of $t_0$ or~$t_1$.
Under this strategy~$\alpha$ we have $d(s_0,s_1) = 0$, i.e., $s_0$ and~$s_1$ are probabilistic bisimilar.
\lipicsEnd
\end{example}